\documentclass[
reprint,
nofootinbib,
 amsmath,amssymb,
 aps,
 pra,
 superscriptaddress,
floatfix,
]{revtex4-2}
\usepackage{listings}
\usepackage{amsthm}
\usepackage{balance}
\usepackage{todonotes}
\usepackage{tikz}
\usepackage{mathtools}

\usepackage{xcolor}
\usepackage{dcolumn}
\usepackage{bm}
\usepackage{bbm}
\usepackage{orcidlink}
\usepackage{physics}
\usepackage[version=4]{mhchem}
\usepackage{gensymb}
\usepackage{textcomp}
\makeatletter
\let\old@makecaption=\@makecaption
\usepackage{subcaption}
\let\@makecaption=\old@makecaption
\makeatother
\usepackage{silence}
\usepackage{graphicx,multirow}
\usepackage{array}
\usepackage{hyperref}
\usepackage[capitalize]{cleveref}
\usepackage[super]{nth}
\usepackage{lipsum}
\usepackage{booktabs}
\AtBeginDocument{\RenewCommandCopy\qty\SI}
\usepackage{dsfont}
\DeclareMathOperator{\Pf}{Pf}
\newcommand{\mcO}{\mathcal{O}}
\newcommand{\mbbC}{\mathbb{C}}

\newtheorem{prop}{Proposition}

\begin{document}
\title{\texorpdfstring{Quantum-Classical Auxiliary-Field Quantum Monte Carlo\\ at the Edge of Practicability}{Quantum-Classical Auxiliary-Field Quantum Monte Carlo at the Edge of Practicability}}
\date{\today}

\author{Francesco Nappi}
\email{francesco.nappi@iqm.tech}
\affiliation{IQM Quantum Computers, Georg-Brauchle-Ring 23-25, 80992, Munich, Germany}
\affiliation{Ludwig Maximilian University of Munich, Geschwister-Scholl-Platz 1, 80539 Munich, Germany}

\author{Matthew Kiser \orcidlink{0000-0002-9357-7583}}
\email{matthew.kiser@iqm.tech}
\affiliation{IQM Quantum Computers, Georg-Brauchle-Ring 23-25, 80992, Munich, Germany}
\affiliation{TUM School of Natural Sciences, Technical University of Munich, Garching, Germany}
\affiliation{Volkswagen AG, Wolfsburg, Germany}

\author{Fedor Šimkovic IV \orcidlink{0000-0003-0637-5244}}
\email{fedor.simkovic@iqm.tech}
\affiliation{IQM Quantum Computers, Georg-Brauchle-Ring 23-25, 80992, Munich, Germany}
\begin{abstract}
We introduce algorithmic improvements to quantum-classical auxiliary-field quantum Monte Carlo (QC-AFQMC) that reduce the dominant per-step classical scaling from $\tilde{\mathcal{O}}(N^{5.5})$ to $\tilde{\mathcal{O}}(N^{4.5})$ as a function of the number of molecular spin-orbitals $N$. Central to this improvement is the application of Aitken's block transformation to handle singular Pfaffians arising in the estimation of overlaps between a quantum trial state and classical Slater-determinant walkers. Together with the use of algorithmic differentiation for the computation of the force bias, this yields a $248\times$ estimated runtime improvement for a system of 100 molecular orbitals. Using our workflow, we demonstrate a ground-state energy calculation for \ce{H8} from quantum data collected on IQM Emerald and post-processed with a tensor-network-based error-mitigation technique. We further validate the method's scalability through noiseless simulation of hydrogen chains up to \ce{H12}, and on the lithium-air battery related rearrangement pathway of the \ce{Li2O4} lithium superoxide dimer in a (26e, 20o) active space. We estimate both quantum and classical runtimes for a potential fault-tolerant implementation of QC-AFQMC, showing that the method holds promise for the early fault-tolerant era. These results move QC-AFQMC a step closer to treating chemically relevant systems. 
\end{abstract}
\maketitle

\section{Introduction}

\begin{figure*}[htbp]
    \centering
    \includegraphics[width=0.92\linewidth]{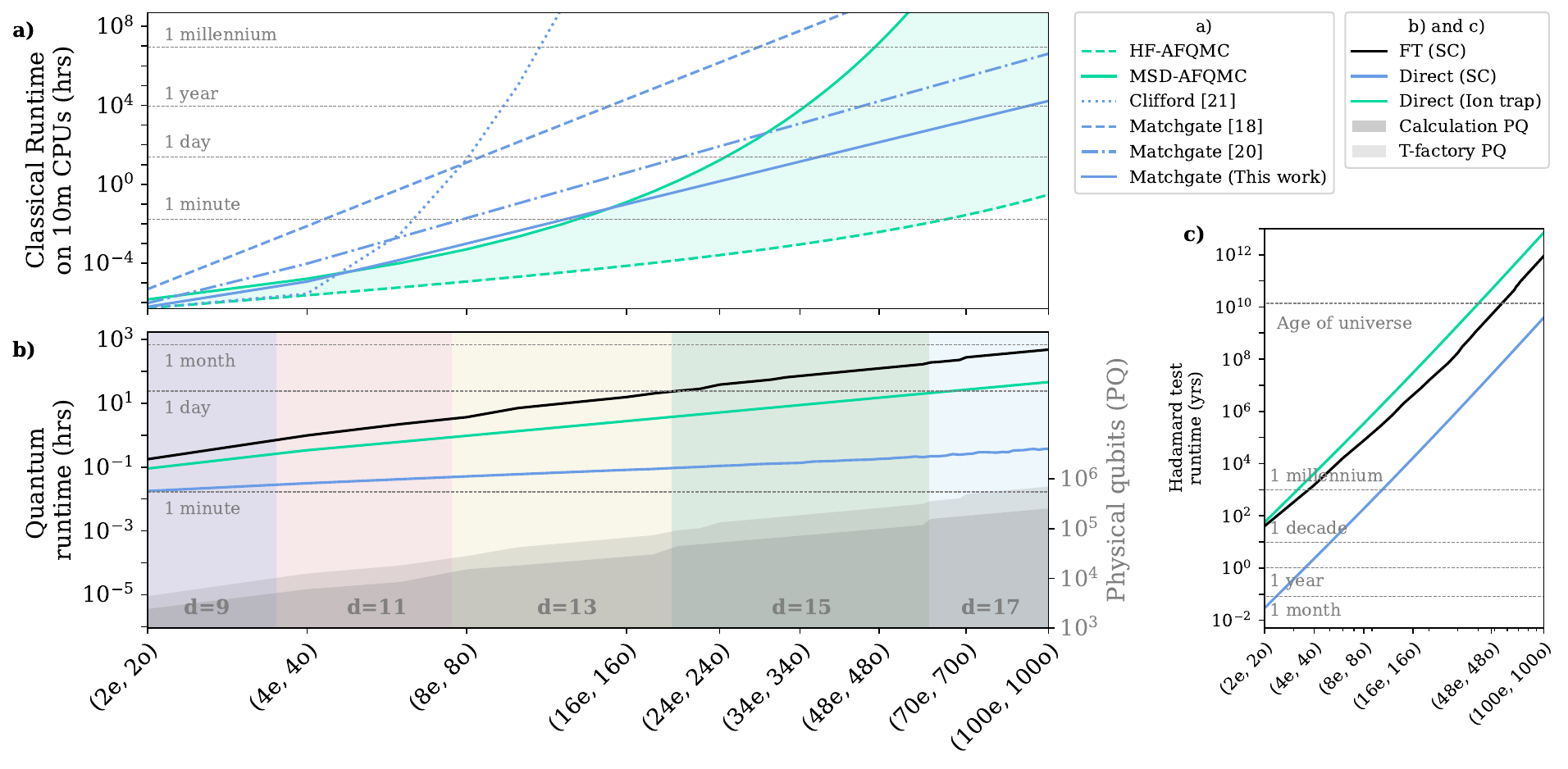}
    \caption{Estimated classical and quantum resource requirements of QC-AFQMC. Further details on assumptions and timings used to construct this figure can be found in \cref{subsec:scaling_assumptions,subsec:quantum_runtime}. \emph{a)} Estimated total classical runtime of classical ph-AFQMC (green) and hybrid QC-AFQMC (blue) implementations as a function of the chemical system size described by $(N_{\text{elec}},\, N_{\text{orb}})$.  \emph{b)} Quantum runtime and physical qubit (PQ) count estimates for collecting matchgate shadow data using a FT quantum computer based on surface code error correction, as well as from direct NISQ implementations using superconducting (SC) and ion-trap quantum hardware. \emph{c)} QC-AFQMC runtimes for NISQ and FT implementations using the Hadamard test to compute overlaps directly on a quantum device.}
    \label{fig:complexity_scaling}
     \stepcounter{subfigure}
    \addtocounter{subfigure}{-1}  
    \refstepcounter{subfigure}\label{fig:complexity_scaling_a}
    \refstepcounter{subfigure}\label{fig:complexity_scaling_b}
    \refstepcounter{subfigure}\label{fig:complexity_scaling_c}
\end{figure*}

Computational chemistry has an impact on numerous industries, including automotive~\cite{narula1996materials}, aviation~\cite{arnold2022materials}, pharmaceuticals~\cite{campos2019importance}, agriculture~\cite{lewer2022structure}, and catalyst design~\cite{hammesschiffer2017catalysts}. Yet, the most important open problems are intractable for transistor-based (classical) computers due to the exponentially growing Hilbert space of fermionic systems. Such problems can, in contrast, be efficiently mapped to quantum computers, and a large part of the search for practical applications of quantum computers has therefore naturally revolved around computational chemistry~\cite{bennett1997strengths,alexeev2024quantumcentric, bärtschi2025potentialapplicationsquantumcomputing, gundlach2025quantumadvantagecomputationalchemistry,huang2025vastworldquantumadvantage}. 

Despite the fact that computing ground states of chemical systems is QMA-hard~\cite{ogorman2022intractability}, it is widely believed that future fault-tolerant (FT) quantum computers using purely quantum algorithms such as quantum phase estimation~\cite{kitaev1995quantum, ni2023low} will deliver a practical advantage over their classical counterparts. This confidence does not extend to near-term (noisy) quantum devices, mainly because the ``first generation" of variational quantum algorithms was shown to suffer from barren plateaus~\cite{mcclean2018barren} as well as prohibitive training and measurement costs~\cite{liu2023variational}. Modern hybrid (quantum-classical) methods promise an alternative path towards practical quantum advantage before the age of fault-tolerance as they distribute the workload between classical and quantum resources according to their relative strengths~\cite{alexeev2025perspective}. They typically use the quantum device only for sampling quantum states while the classical device performs the algorithmic heavy-lifting based on the resulting quantum data. 

A hybrid algorithm that has attracted considerable attention from both academia and industry~\cite{amsler2023classical,huang2024evaluating,kiser2025contextual,zhao2025quantum} is quantum-classical auxiliary-field quantum Monte Carlo (QC-AFQMC)~\cite{huggins2022unbiasing}, which belongs to the broader family of hybrid QMC methods that combine quantum state preparation with classical Monte Carlo diffusion~\cite{jiang2025walking,buonaiuto2026unified}. At its core, the approach builds on classical phaseless AFQMC (ph-AFQMC), a ground-state projector quantum Monte Carlo method that represents electron–electron interactions with auxiliary fields and stochastically propagates the system in imaginary time via an ensemble of Slater determinant (SD) walkers. The ``phaseless” label refers to an approximate constraint on the walker evolution imposed through a trial wavefunction, and aimed at mitigating an exponential growth of statistical noise that would otherwise arise from uncontrolled complex phases in the walker weights~\cite{zhang2013auxiliary}. Over the past two decades, ph-AFQMC has established itself as a workhorse algorithm in condensed matter physics~\cite{xu2024coexistence} as well as quantum chemistry~\cite{lee2022twenty}, delivering excellent accuracy at low polynomial-order computational cost.

The main limitation of ph-AFQMC comes from the critical dependence on the trial wavefunction, which introduces a systematic, non-variational bias that vanishes only in the limit of using the true ground state as the trial state. Trial wavefunctions are typically chosen to be either single- (i.e. Hartree-Fock) or multi-SD (MSD) representations obtained from other numerical methods, such as density matrix renormalization group (DMRG)~\cite{jiang2025unbiasing}, coupled-cluster singles and doubles (CCSD)~\cite{kjonstad2025systematic} or selected configuration interaction (SCI)~\cite{mahajan2022selected}. For strongly correlated systems, however, the number of required SDs can scale exponentially with the system size, limiting the practical applicability of MSD-AFQMC for many systems of interest~\cite{zweig2022towards}. 

QC-AFQMC addresses this limitation by using wavefunctions prepared on a quantum device as trial states, with the expectation that quantum computers can more efficiently prepare trial states of higher fidelity than classical methods. 
The central computational bottleneck of QC-AFQMC is the calculation of overlaps between the quantum trial wavefunction and classical SD Monte Carlo (MC) walkers. Although these can, in principle, be computed directly on quantum devices via a modified Hadamard test~\cite{luongo2024quantum}, the sheer number of required overlaps during a typical ph-AFQMC run renders this option computationally infeasible at the current availability and clock speed of quantum devices~\cite{kiser2024classical}. 

For this reason, it was originally proposed~\cite{huggins2022unbiasing} to instead evaluate overlaps using classical shadows~\cite{huang2020predicting}, which has the advantage of collecting quantum data only once and then offloading all overlap calculations to a classical post-processing step. Unfortunately, the initial use of Clifford shadows scaled exponentially with the system size $N$~\cite{huggins2022unbiasing}. Using matchgate shadows instead allowed a polynomial classical runtime scaling as $\tilde{\mathcal{O}}(N^{8.5})$ for each time step of QC-AFQMC with a variance scaling as $\mathcal{O}(\sqrt{N}\log(N))$ for each overlap evaluation~\cite{wan2023matchgate}. 
The use of derivatives-based expressions within the evaluation of the force bias and local energy subsequently reduced the exponent down to $\tilde{\mathcal{O}}(N^{5.5})$~\cite{jiang2025unbiasing}. Yet, it remains prohibitively high for systems of scientific interest, as can be seen from \cref{fig:complexity_scaling_a} (dot-dashed blue line) and is significantly worse than the $\mathcal{O}(N^{4})$ scaling of the purely-classical ph-AFQMC algorithm using single-SD trial states~\cite{zhang2013auxiliary}. 

In this work, we further reduce the scaling of the classical post-processing to $\tilde{\mathcal{O}}(N^{4.5})$ per time step through algorithmic improvements in the matchgate-shadow overlap calculation, thereby significantly decreasing the classical runtime of QC-AFQMC. Using our accelerated implementation, we study the numerical stability and convergence properties of QC-AFQMC beyond the current state-of-the-art. Specifically, we compute the ground-state energies of hydrogen chains: $\ce{H8}$ using experimental data from IQM Emerald \cite{abdurakhimov2024technologyperformancebenchmarksiqms}, and \ce{H12} using simulated data. Further, we study three configurations of the rearrangement pathway of $\ce{Li2O4}$ lithium superoxide dimer (see \cref{fig:li_superoxide_pathway}) using 40-qubit simulated data within a (26e, 20o) active space. This pathway is directly relevant to the charge and discharge cycles of lithium–air batteries, which have attracted considerable scientific interest owing to their exceptionally high theoretical energy density, comparable to that of hydrocarbon fuels and substantially greater than conventional lithium-ion batteries~\cite{bryantsev2010stability,badwal2014emerging,das2014structure,kale2024comprehensive,gao2021computational}. 

We also investigate the quantum resource requirements of QC-AFQMC (shown in \cref{fig:complexity_scaling_b,fig:complexity_scaling_c}) and find that, unlike using the Hadamard test~\cite{zhang2025quantum}, the matchgate shadow protocol yields reasonable quantum runtimes for both the near-term intermediate-scale quantum (NISQ)~\cite{preskill2018quantum} and FT quantum computing regimes. Taken together, the results presented in \cref{fig:complexity_scaling} suggest that QC-AFQMC is approaching the boundary of practical feasibility for investigations of chemically relevant systems.

The paper is structured as follows. In \cref{sec:results}, we outline the main algorithmic improvements with respect to the classical (\cref{subsec:classical_runtime}) and quantum (\cref{subsec:quantum_runtime}) scaling of QC-AFQMC and present the results of our calculations performed using experimental and simulated data (\cref{subsec:qcafqmc_implementation}). In \cref{sec:discussion}, we discuss the implications of our results in the broader context of QC-AFQMC and suggest future research directions. Detailed descriptions of ph-AFQMC, QC-AFQMC, the error mitigation methods used for our quantum data, the assumptions for our classical and quantum runtime estimates, and our improvements of the matchgate shadow post-processing steps within QC-AFQMC are found in \cref{sec:methods}.

\begin{figure}
    \centering
    \includegraphics[width=0.92\linewidth]{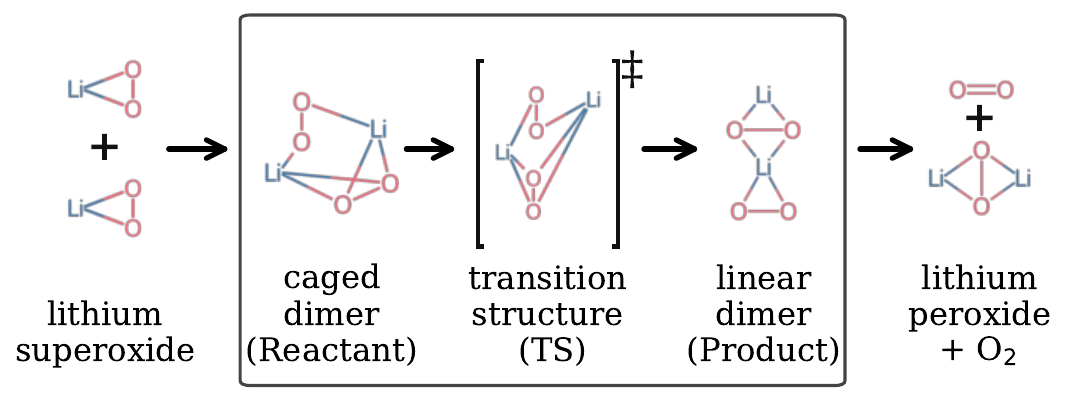}
    \caption{Reaction pathway of lithium superoxide to lithium peroxide and molecular oxygen \cite{bryantsev2010stability}.
    }
\label{fig:li_superoxide_pathway}
\end{figure}

\section{Results}\label{sec:results}

\subsection{Classical runtime}\label{subsec:classical_runtime}

This section introduces the ph-AFQMC algorithm, derives the key quantities required for its quantum-classical variant (QC-AFQMC), outlines the computational bottleneck in matchgate shadow post-processing, and presents a method that reduces its scaling complexity by one polynomial degree. 

The ph-AFQMC algorithm uses a reference trial state $|\psi_{\rm T}\rangle$ to guide the stochastic, Trotterized, imaginary-time evolution 
\begin{equation}
\left|\Psi_{\text{GS}}\right\rangle \propto \lim _{n\rightarrow \infty}\left(\exp 
 (-\Delta \tau \hat{H})\right)^n|\Phi_0\rangle\,,
\end{equation}
of an initial state $|\Phi_0\rangle$ expressed as an ensemble of SD walkers $\ket{\varphi_i}$ to estimate the ground state $\ket{\Psi_{\text{GS}}}$ of a second-quantized Hamiltonian of the form:
\begin{equation}
\begin{aligned}\label{eq:qcham}
\hat{H}&= \hat{H}_1 + \hat{H}_2 \\
&=\sum_{p,q=1}^N h_{p q} \hat{a}_p^{\dagger} \hat{a}_q^{\phantom{\dagger}}+\frac{1}{2} \sum_{\gamma=1}^{N_{\rm C}} \sum_{p,q,r,s=1}^N  \!\! L_{ps}^\gamma L_{qr}^{\gamma}\hat{a}_p^{\dagger} \hat{a}_q^{\dagger} \hat{a}_r^{\phantom{\dagger}} \hat{a}_s^{\phantom{\dagger}}\,,
\end{aligned}
\end{equation}
where $h_{pq}$ are one-electron integrals, the two-electron interaction is represented in a compact, Cholesky-decomposed form and $L^\gamma_{ps}$ are the $N_{\rm C}$-many Cholesky matrices of the electron-repulsion integral tensor.
To preserve the SD nature of the walkers, all non-quadratic terms in the imaginary-time evolution are decomposed into integrals over auxiliary fields via a Hubbard-Stratonovich transformation~\cite{hubbard1959calculation,*stratonovich1957method,negele2018quantum}. These integrals are then evaluated through a stochastic Monte Carlo sampling process that uses $|\psi_{\rm T}\rangle$ for importance sampling and the mitigation of the underlying phase problem~\cite{zhang2003quantum}. 

In QC-AFQMC~\cite{huggins2022unbiasing}, the trial state $|\psi_{\rm T}\rangle$ is prepared on a quantum computer; thus, the central intersection between the quantum and classical parts of QC-AFQMC is the calculation of the overlap between this trial state and the SD MC walkers. This is done using matchgate shadow tomography~\cite{wan2023matchgate}. In this approach, one collects many measurements (``snapshots") of the state in different effective bases by sampling from the distribution of matchgate circuits whose associated inverse channels can be computed efficiently on a classical computer. This post-processing step produces unbiased estimates of properties of the state, in this case state overlaps. More detailed treatments of QC-AFQMC and the overlap reconstruction from matchgate shadows are found in \cref{subsec:qcafqmc_description} and \cref{subsec:calc_ovlp_mg_shadows}, respectively.

In QC-AFQMC, there are three quantities that must be computed from overlap estimates between the quantum trial state and SD walkers using matchgate shadow data: the Monte Carlo weight $w_i(\tau)$ and force bias for each walker 
$\ket{\varphi_i(\tau)}$ 
at each time step $\tau$, 
and the local energy of each walker at periodic points throughout the evolution.
For the $i$-th walker, its weight $w_i(\tau)$ (\cref{eq:phaseless}) depends on the ratio of overlaps between the trial state $\ket{\psi_{\rm T}}$ and the walker state at both the start and end of the corresponding time step (\cref{eq:ovlp_ratio}). The $\gamma$-th term of the force bias $\bar{\mathbf{x}}\in \mathbb{C}^{N_{\rm C}}$
can be expressed as~\cite{jiang2025unbiasing}
\begin{equation}
    \bar{x}_\gamma=-\sqrt{\Delta\tau}\left.\frac{\partial\ln(\bra{\psi_{\rm T}}e^{\lambda_\gamma\hat{v}_\gamma}\ket{\varphi_i(\tau)})}{\partial\lambda_\gamma}\right|_{\lambda_\gamma=0}\,,\label{eq:force_bias}
\end{equation}
where  $\hat{v}_{\gamma}=i\sum_{pq}L^{\gamma}_{pq}\hat{a}^{\dagger}_{p}\hat{a}_{q}$ is the one-body Cholesky operator. 
By Thouless' theorem~\cite{thouless1960stability,thouless1961vibrational}, the action of $e^{\lambda_\gamma\hat{v}_\gamma}$ on the SD $\ket{\varphi_i(\tau)}$ yields another SD, and \cref{eq:force_bias} can be written as a ratio of overlaps (see \cref{subsec:fb_improvement}). 

The local energy $E_{\rm L}$ with respect to a walker at a given QC-AFQMC step can similarly be obtained from the sum of its one- and two-body parts~\cite{jiang2025unbiasing}
\begin{align}
    E_1 &= \left.
    \frac{\partial \ln \bigl(
    \bra{\psi_{\rm T}} e^{\lambda \hat{H}_1} \ket{\varphi_i}
    \bigr)}
    {\partial \lambda}
    \right|_{\lambda=0}
    \label{eq:h1_energy_diff}\\
    E_2 &= \frac{1}{2 \braket{\psi_{\rm T}}{\varphi_i}}
    \sum_\gamma
    \left.
    \frac{\partial^2
    \bra{\psi_{\rm T}}
    e^{-i\lambda_\gamma \hat{v}_\gamma}
    \ket{\varphi_i}}
    {\partial \lambda_\gamma^2}
    \right|_{\lambda_\gamma=0}\,.\nonumber
\end{align}
Both contributions are expressed entirely as derivatives of overlaps, meaning that the local energy can also be obtained solely from estimating overlaps via matchgate shadows.

Let us consider the classical scaling of evaluating the aforementioned quantities. The computation of MC walker weights corresponds to ratios of single overlap estimates, which scale as $\mathcal{O}(N^4)$ per snapshot in the matchgate shadows protocol of Ref.~\onlinecite{wan2023matchgate}. It has been shown that the evaluation of the force bias per snapshot scales similarly to a single overlap calculation $\mathcal{O}(N^4)$, with the post-processing of the local energy scaling as $\mathcal{O}(N^{5})$~\cite{jiang2025unbiasing}. 

The computational bottleneck in estimating an overlap with matchgate shadows lies in evaluating the polynomial
\begin{equation}
    q(z)\sim \Pf(A(z)) \equiv \Pf(B+zC)\,,
\end{equation}
where the matrices $B$ and $C$ inside the Pfaffian are square, skew-symmetric and of dimension $2N-\zeta$, and $\zeta$ is the number of fermions in the system (see \cref{subsec:calc_ovlp_mg_shadows} for more details). When $B$ is invertible, the coefficients of $q(z)$ can be computed in $\mathcal{O}(N^3)$ via the differentiation method of Ref.~\onlinecite{wan2023matchgate}, summarized in \cref{app:appDalg}. In the setting of QC-AFQMC, $C$ typically has full rank, while $B$ comes with the following structure
\begin{align}\label{eq:def_B}
    B\equiv&\bigoplus_{j=1}^{\zeta/2}
    \begin{bmatrix}
        0&0\\
        0&0
    \end{bmatrix}
    \oplus
    \bigoplus_{j=1}^{N-\zeta}
    \begin{bmatrix}
        0&1\\
        -1&0
    \end{bmatrix}
    \equiv
    \begin{bmatrix}
        \mathbf{0}&\mathbf{0}\\
        \mathbf{0}&J
    \end{bmatrix}\,,
\end{align}
meaning that $B$ is never invertible by construction. Therefore, the coefficients of $q(z)$ must be computed through interpolation, which scales as $\mathcal{O}(N^4)$ due to the $\mathcal{O}(N)$-many required calculations of Pfaffians which take $\mathcal{O}(N^3)$ time~\cite{wimmer2012algorithm}. Consequently, the existing per-snapshot overlap evaluation using matchgate shadows scales as $\mathcal{O}(N^4)$~\cite{wan2023matchgate}. In the following, we show how to reduce this cost by one polynomial degree to $\mathcal{O}(N^3)$.

The central realization is that the skew-symmetric matrix $A(z)$ can be partitioned as
\begin{align}\label{eq:split_A}
    A(z)=
    \begin{bmatrix}
        zC_{11}&zC_{12}\\
        -zC_{12}^{\rm T}&J + zC_{22}
    \end{bmatrix}\,,
\end{align}
where $C$ is split into four blocks with $C_{21} = - C_{12}^{\rm T}$, and such that $C_{22}$ is of the same dimensions as $J$ in \cref{eq:def_B}. Applying Aitken's block transformation formula for Pfaffians~\cite{bunch1982note,zhang2006schur,gonzalez2011numeric} yields
\begin{equation}
\begin{aligned}
    \Pf(A)=&\Pf(zC_{11})\Pf(J+zC_{22}+zC_{12}^{\rm T}C_{11}^{-1}C_{12})\\
    =&z^{\zeta/2}\Pf(C_{11})\Pf(J+zF)\,,
\end{aligned}
\end{equation}
where $F=C_{22}+C_{12}^{\rm T}C_{11}^{-1}C_{12}$. Since $J$ is non-singular, the derivatives of $\Pf(J+zF)$ can be computed using the differentiation method of Ref.~\cite{wan2023matchgate} as a subroutine, and the higher-order derivatives of $\Pf(J+zF)$ follow by the Leibniz rule. As shown in \cref{app:symmetry}, the coefficients of $q(z)$ satisfy a symmetry that reduces the computational cost by an additional factor of two. 
The asymptotic scaling is dominated by an eigenvalue problem and a single Pfaffian calculation, both of which scale as $\mathcal{O}(N^3)$. Compared with the interpolation approach, our method improves the numerical accuracy of estimating $q(z)$, most notably near $z\sim0$; this is discussed further in \cref{app:num_stability}. 

A second $\mathcal{O}(N^3)$ method for the computation of overlaps is presented in \cref{app:skew_method}, where the covariance matrix of the vacuum state is perturbed by $\pm\beta$ along an auxiliary variable and the true coefficients are recovered through interpolation. In practice, this approach proved less numerically stable than the Aitken-based method.

Crucially, this $\mathcal{O}(N^3)$ overlap post-processing directly improves the scaling of both the force bias and the local energy. In \cref{subsec:fb_improvement}, we demonstrate how our improved algorithm yields a more efficient force-bias computation than prior state-of-the-art, maintaining the same asymptotic computational cost of a single overlap estimation~\cite{jiang2025unbiasing}. The local energy now scales as $\mathcal{O}(N^3N_{\rm C})$ per snapshot and, since $N_{\rm C}$ grows linearly with system size~\cite{motta2018ab}, the matchgate post-processing for local energy calculations scales overall as $\mathcal{O}(N^4)$.

Altogether, our improvements reduce the asymptotic scaling of QC-AFQMC from $\tilde{\mathcal{O}}(N^{7.5})$~\cite{jiang2025unbiasing} to $\tilde{\mathcal{O}}(N^{6.5})$, given the $\mathcal{O}(N^{4})$ per-snapshot, local-energy cost, the requirement of $\mathcal{O}(\sqrt{N}\log(N))$ snapshots for an overlap estimate~\cite{wan2023matchgate}, and the $\mathcal{O}(N^{2})$ local energy evaluations required for the final QC-AFQMC energy estimate~\cite{lee2022twenty}. \cref{fig:complexity_scaling_a} compares estimated runtimes of our QC-AFQMC implementation (solid blue) against prior work~\cite{huggins2022unbiasing,wan2023matchgate,jiang2025unbiasing,zhao2025quantum} (dotted/dashed blue) for various system sizes, assuming ten million available CPU cores (more details in \cref{subsec:scaling_assumptions}). Our method outperforms all prior implementations. Notably, the original algorithm of Ref.~\onlinecite{huggins2022unbiasing} using a Clifford shadow protocol scales prohibitively beyond around 20 qubits, well within the realm of full configuration-interaction (FCI) calculations. Similarly, the best prior matchgate implementation~\cite{zhao2025quantum} would require multiple years for a 50-orbital calculation, a system size considered beyond classical brute-force methods, and close to half a millennium for 100 orbitals. In contrast, our implementation requires approximately one week and 1.8 years for 50 and 100 orbitals, respectively, corresponding to relative improvements of up to $248\times$.  
We note that all runtime estimates for our improved matchgate shadows protocol have been obtained by extrapolating from our largest (40-qubit) computation performed on the LUMI supercomputer~\cite{lumi2026supercomputer}.

\cref{fig:complexity_scaling_a} also shows the classical scaling of ph-AFQMC (green lines and shaded area). A single-determinant Hartree-Fock (HF) trial (dashed green) is extremely fast but may not be sufficiently accurate to capture the properties of the ground-state wavefunction. If the fidelity of the classical trial state is required to remain constant as the system size grows, the number of SDs in its description can scale exponentially for strongly correlated systems. Adopting the conservative estimate of $2^{0.37N}$ SDs~\cite{kanno2023quantumselectedconfigurationinteractionclassical}, the computational runtime of multi-Slater-determinant MSD-AFQMC (full green line) grows very rapidly and intersects that of QC-AFQMC at approximately 15 orbitals. Beyond this crossover, it may become advantageous to encode the MSD trial state on a quantum rather than classical device, which will likely require the use of FT quantum computers. We assess the quantum runtime of such QC-AFQMC implementations in the next section.

\subsection{Quantum runtime}\label{subsec:quantum_runtime}

The quantum runtime of QC-AFQMC in the NISQ and FT settings has received little attention beyond a brief discussion in Ref.~\cite{kiser2024classical}. We fill this gap by estimating the quantum runtime in both settings, as shown in \cref{fig:complexity_scaling_b}. We additionally analyze the runtime of the modified Hadamard test~\cite{luongo2024quantum} approach to estimating overlaps directly on a quantum device in \cref{fig:complexity_scaling_c}. Further details on the underlying assumptions of our estimates are found in \cref{subsec:scaling_assumptions_quantum}. 

We begin by estimating the circuit execution time of QC-AFQMC. 
For this, we consider a randomly-initiated local unitary cluster Jastrow (LUCJ) ansatz state~\cite{matsuzawa2020jastrow} with a single ansatz block repetition and square grid connectivity. Such circuits are considered likely classically difficult
to sample from, given that the related UCJ circuits generalize IQP circuits \cite{hafid2025hardness}. Following the shadow protocol, a matchgate circuit sampled from the orthogonal group $O(2N)$ is appended to the ansatz preparation circuit. Using gate execution timings of superconducting (SC, blue) and ion-trap (IT, green) quantum processors, we found execution times to remain reasonable for large system sizes, i.e.  roughly 21 minutes (SC) and 1.8 days (IT) for 100 orbitals calculations. 

Next, we estimate the FT execution time for a superconducting device with square connectivity, using surface codes with their code distances (indicated by the shaded regions in \cref{fig:complexity_scaling_b}) adapted to maintain a constant state fidelity. Under conservative physical fidelity estimates, a surface code of at least distance $d=9$ is required even for small systems, the distance growing to $d=15$ at 50 orbitals and $d=17$ at 100 orbitals. This corresponds to a requirement of $104$ and $255$ thousand physical qubits for the trial state with an additional $169$ and $450$ thousand physical qubits required for T-gate factories, respectively. 
In terms of execution time, we found that approximately 19 days are required for the data collection of 100-orbital systems in the FT setting, which is significant but still far below the corresponding classical runtime of QC-AFQMC.

Overlaps can also be calculated using the modified Hadamard test~\cite{luongo2024quantum} rather than matchgate shadow tomography. Given the circuit depths and controlled-unitary operations, this approach was previously deferred to the FT regime and supplanted by shadow tomography~\cite{huggins2022unbiasing,kiser2024classical}. In \cref{fig:complexity_scaling_c}, we place more concrete numbers behind this assessment. The dominant overhead is the repeated circuit execution required to estimate measurement averages for all MC walkers and at all time-steps. We find that a system of just 4 orbitals would require a millennium of compute time; at 100 orbitals, the estimate exceeds the age of the universe. This conclusion is not materially improved by removing the FT infrastructure overhead from the circuit model. We thus conclude that the modified Hadamard test is not a viable path for QC-AFQMC at any practically relevant system size with current or near-term methods unless significant algorithmic improvements are achieved.

\subsection{Implementation and scaling}\label{subsec:qcafqmc_implementation} 

We now proceed to investigate the performance of our improved QC-AFQMC workflow on various classes of quantum chemistry systems. 

First, we study stretched hydrogen chains up to \ce{H12} at $R=1.5\text{\AA}$. We use the separable pair approximation (SPA) ansatz~\cite{kottmann2022optimized} for trial state preparation. SPA states are easy to simulate and train classically, and the corresponding molecular orbital basis set is determined as part of the SPA optimization procedure~\cite{kottmann2022optimized}. The SPA ansatz is physically-inspired and hardware-efficient, meaning that it generates quantum states with high overlap with the ground state while being implementable through short quantum circuits. 

In \cref{fig:evolution_plot_h8_H12}, we present QC-AFQMC evolution curves using error-mitigated measurements on IQM's Emerald quantum computer~\cite{abdurakhimov2024technologyperformancebenchmarksiqms} for \ce{H8} (\cref{fig:evolution_plot_h8_H12_a}) as well as noiseless simulated measurements for \ce{H12} (\cref{fig:evolution_plot_h8_H12_b}). We show curves of three different AFQMC calculations for comparison. For two of the three, the walker propagation used matchgate shadow data to estimate their weights and the force bias, but they differ in how the local energy was computed: one via matchgate shadow data (MG-QC-AFQMC, blue) and the other via the state vector representation of the trial state (Exact-QC-AFQMC, green). The third curve is a fully classical ph-AFQMC calculation (yellow). This allows us to compare the stability of QC-AFQMC under the statistical noise of the matchgate protocol, and the performance of calculating the local energies $E_{\rm L}$ and the total energy at the end of the evolution. More details of the calculations can be found in \cref{subsec:implementation_details}.

\begin{figure}
    \centering
    \includegraphics[width=0.92\linewidth]{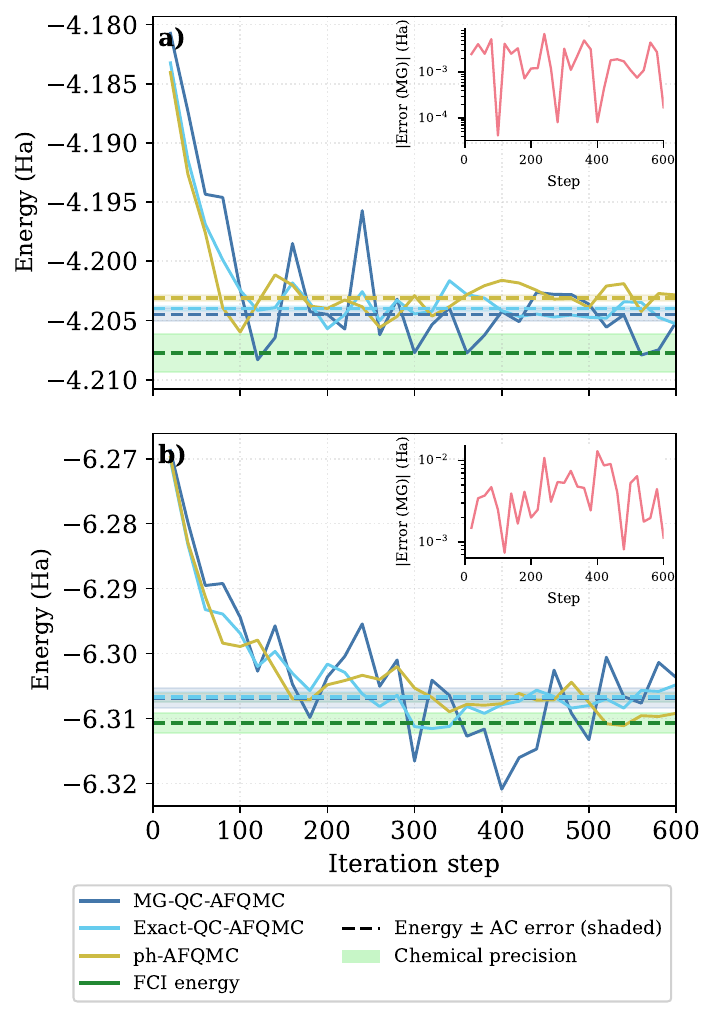}
    \caption{QC-AFQMC evolution for a) \ce{H8} and b) \ce{H12} hydrogen chains with $R=1.5\text{\AA}$ using experimental and simulated data, respectively. We show convergence curves where energies are computed using either matchgate shadow data (MG-QC-AFQMC) or via a state-vector representation of the trial state (Exact-QC-AFQMC). Fully classical ph-AFQMC curves and exact FCI energies are shown for comparison.    Insets show the error of the energy estimation of MG-QC-AFQMC with respect to the exact energy calculation of Exact-QC-AFQMC.}
    \label{fig:evolution_plot_h8_H12}
    \stepcounter{subfigure}
    \addtocounter{subfigure}{-1}  
    \refstepcounter{subfigure}\label{fig:evolution_plot_h8_H12_a}
    \refstepcounter{subfigure}\label{fig:evolution_plot_h8_H12_b}
\end{figure}

\begin{figure}
    \centering
\includegraphics[width=0.92\linewidth]{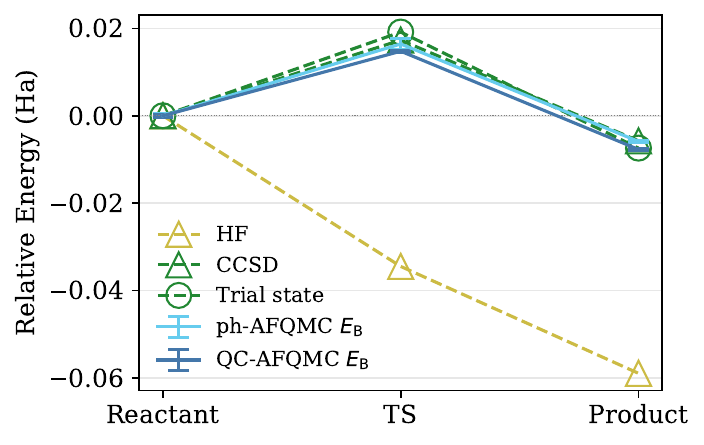}
    \caption{Reaction energies for three configurations of the rearrangement pathway of \ce{Li2O4} lithium superoxide dimer (shown in \cref{fig:li_superoxide_pathway}), where lithium peroxide and oxygen (product) are formed from a lithium superoxide dimer (reactant) through a transition state (TS).
    Calculations are performed in the cc-pVDZ basis set using an active space of 26 electrons in 20 orbitals. We compare the results for Hartree-Fock, CCSD, the trial state (top 3 determinants of CCSD) and block energies $E_{\rm B}$ during QC-AFQMC and ph-AFQMC calculations. Absolute energies are given in \cref{tab:energies}.}
    \label{fig:qc_afqmc_performance_40qb}
\end{figure}

The energies calculated in the three methods align in both chemical systems, albeit with the MG-QC-AFQMC calculation exhibiting a somewhat higher variance (See \cref{tab:hchain_energies} in \cref{app:abs_energies}). For both \ce{H8} and \ce{H12}, the energies of all three AFQMC variants are slightly higher than the FCI energies, which can be attributed to the bias coming from the trial state, as evidenced by the results of the ph-AFQMC calculation.

Next, we turn our attention to the three configurations along the \ce{Li2O4} rearrangement pathway shown in \cref{fig:li_superoxide_pathway}: the caged dimer (reactant), transition structure (TS) and linear dimer (product) using the geometries from Ref.~\cite{gao2021computational}. We solve all configurations within an active space obtained from a restricted Hartree–Fock (RHF) calculation in the cc-pVDZ basis set. Here, the six lowest orbitals (two \ce{Li}~$1s$ and four \ce{O}~$1s$) were frozen, leaving an active space of 26 electrons (13 spin-up and 13 spin-down) in 20 orbitals, corresponding to 40 qubits. For each configuration, we prepared trial states by mapping the CCSD wavefunction onto a configuration-interaction expansion and retaining the three highest-weight determinants. This enables the efficient simulation of the 40-qubit circuits~\cite{dias2024classical} while also recovering the correct qualitative behavior for the reaction path within ph-AFQMC. More details on the trial state and sampling method are provided in \cref{app:large_sampling}.

\cref{fig:qc_afqmc_performance_40qb} shows reaction energies for the pathway with the absolute energies given in \cref{tab:energies}. Instead of calculating a full QC-AFQMC energy, we perform a test by using matchgate shadows for the evolution and compute the block energy $E_{\rm B}$ at the 100th time step, which we compare to the same test using ph-AFQMC (see \cref{subsec:implementation_details}). The reaction energies estimated from the block energies of QC-AFQMC and ph-AFQMC are in the neighborhood of the CCSD calculations and recover the expected triangular shape of the reaction pathway based on prior literature~\cite{bryantsev2010stability,das2014structure,gao2021computational}. While these results were obtained from a single block energy evaluation, it implies that for a full QC-AFQMC calculation with more walkers and $\mathcal{O}(N^2)$ energy evaluations one could reasonably expect to recover the nature of the reaction pathway, thus demonstrating the algorithmic viability at previously inaccessible scales.

\begin{figure}
    \centering
    \includegraphics[width=0.92\linewidth]{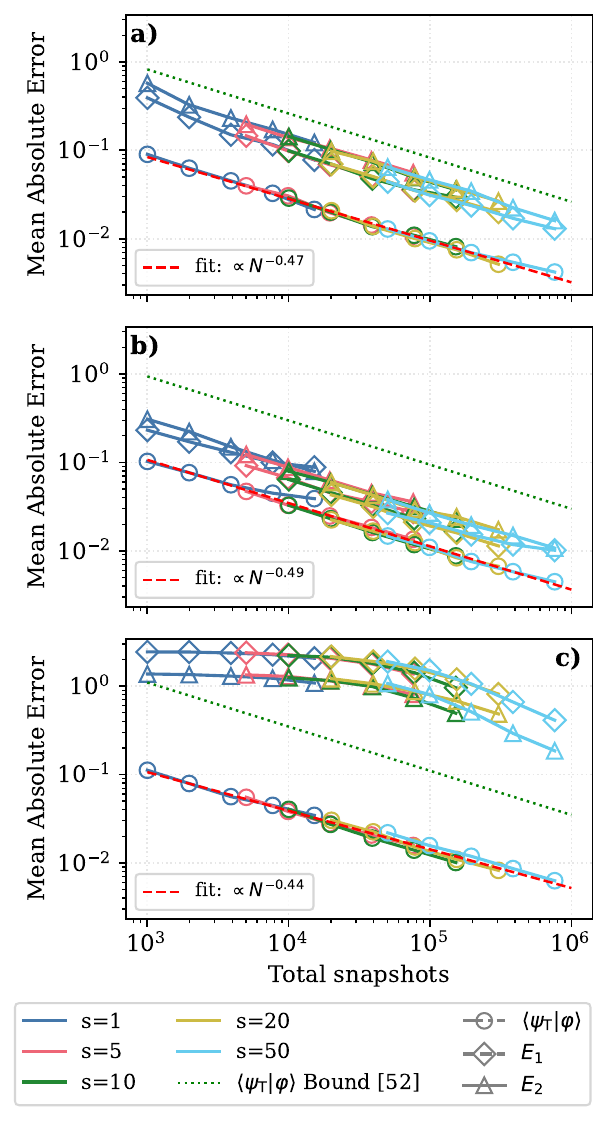}
    \caption{Overlap, one-body energy, and two-body energy estimates for a) \ce{H8} (16 qubits), b) \ce{H12} (24 qubits), and c) \ce{Li2O4} (40 qubits) using simulated matchgate shadows with a varying number of shots per shadow unitary. Also shown are the fit lines for the overlap scaling (red dashed) as well as upper bounds on the error of the overlap calculation based on the number of samples in \cref{eq:n_samples} (green dotted)~\cite{wan2023matchgate}.}
    \label{fig:snapshot_scaling}
    \stepcounter{subfigure}
    \addtocounter{subfigure}{-1}  
    \refstepcounter{subfigure}\label{fig:snapshot_scaling_a}
    \refstepcounter{subfigure}\label{fig:snapshot_scaling_b}
    \refstepcounter{subfigure}\label{fig:snapshot_scaling_c}
\end{figure}

Let us now focus on a more detailed analysis of the performance of our matchgate shadows implementation used for the computations of \cref{fig:evolution_plot_h8_H12,fig:qc_afqmc_performance_40qb}. In \cref{fig:snapshot_scaling}, we report the mean absolute error of the estimated overlap between the trial state and a randomly chosen MC walker as well as its one- and two-body local energy estimates. Results are shown for a) \ce{H8}, b) \ce{H12}, and c) \ce{Li2O4} in the TS state using noiseless matchgate shadow data. The scaling of the overlap error with respect to the number of snapshots $N_{\rm s}$ follows the expected $\mathcal{O}(N_{\rm s}^{-0.5})$ behavior and remains below its theoretical bound~\cite{wan2023matchgate}. While this is expected for overlap estimation, no analogous guarantees exist for the energy estimators. These involve finite-difference derivatives, requiring differences of overlap estimates, which can amplify statistical fluctuations. Furthermore, the evaluation of the two-body energy involves an additional summation over the Cholesky index $\gamma$, introducing another source of variance. For \ce{Li2O4}, we employ a median-of-means estimator with a fixed number of bins, which accounts for the slower convergence observed at low numbers of snapshots.

Collecting multiple shots per matchgate unitary can improve the quantum workload in some architectures~\cite{helsen2023thrifty, zhou2023performance}. However, a study on the utility of multi-shot matchgate shadows for QC-AFQMC has been so far missing. For small systems $N\leq8$, we see diminishing returns of using additional shots per unitary because the number of possible measurement outcomes is small relative to the shot count. For larger systems, the number of potential outcomes far exceeds any practical number of shots. Thus, as seen in \cref{fig:snapshot_scaling}, error scaling is shown to be unaffected by shot multiplicity, implying an equivalent tradeoff between unitaries and shots at larger system sizes.

The matchgate shadows protocol for calculating ratios has been shown to be inherently noise resilient to Markovian, invertible, gate-independent quantum error channels, because the rescaling factors from the errors cancel in the ratios~\cite{chen2021robust,koh2022classical,zhao2024grouptheoretic,huang2024evaluating,zhao2025quantum}. However, mitigation is required for the other types of noise. We therefore employ three complementary error mitigation techniques. 

First, we exploit parity constraints appearing inherently within the matchgate shadows protocol. As mentioned in \cref{subsec:classical_runtime}, matchgate shadow tomography starts with sampling from the orthogonal group $O(2N)$. The determinant of the orthogonal matrix $Q\in O(2N)$ governs whether the parity of the state is preserved ($\det(Q)=+1$) or flipped ($\det(Q)=-1$). Thus, measurements inconsistent with the expected parity can be discarded.

Second, we employed the robust matchgate shadows method~\cite{chen2021robust, koh2022classical, zhao2024grouptheoretic, wu2024errormitigated}, which approximately learns the impact of the noise on the matchgate channel; we discuss this further in \cref{app:robust}. Although it improves the overlap estimates, we found it to have limited impact on the calculation of ratios of overlaps. This conclusion is in line with what was observed in Ref.~\onlinecite{zhao2025quantum}.

\begin{figure}
    \centering
    \includegraphics[width=0.92\linewidth]{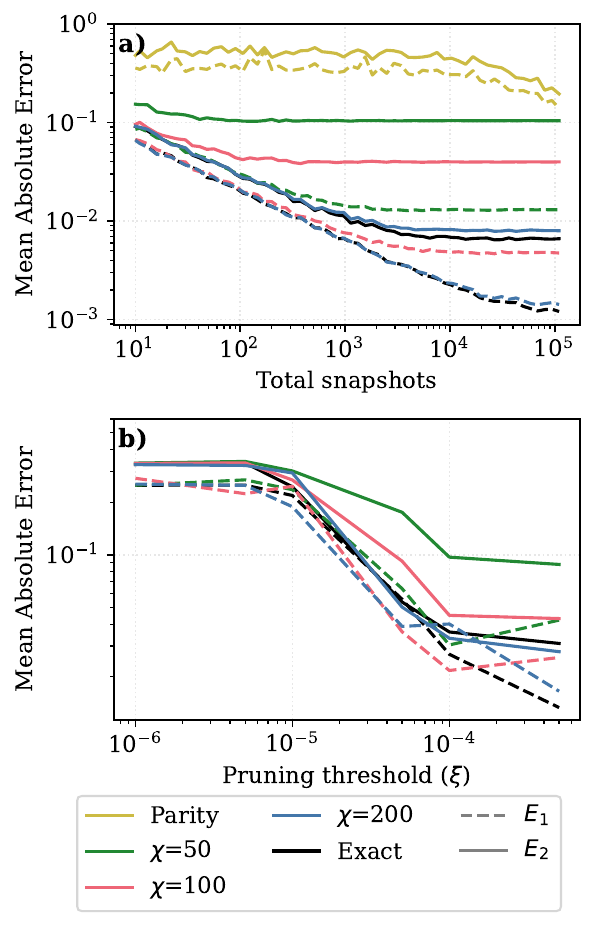}
    \caption{Performance analysis of the filter state method using MPS states for the \ce{H8} hydrogen-chain system. \emph{a)} Overlap, one-body energy and two-body energy errors for a fixed filter threshold $\xi$ and variable bond dimensions $\chi$ of the classical simulation of the circuit are used as a filter. SV corresponds to a full state vector filter representation. \emph{b)} The number of snapshots is fixed to 2,000, while $\chi$ and $\xi$ are varied.}
    \label{fig:mps_post_selection}
    \stepcounter{subfigure}
    \addtocounter{subfigure}{-1}  
    \refstepcounter{subfigure}\label{fig:mps_post_selection_a}
    \refstepcounter{subfigure}\label{fig:mps_post_selection_b}
\end{figure}

Third, we developed a filtering heuristic, where post-selection is based on filter states constructed using a classically efficient representation of the circuits---comprising trial state preparation with the matchgate unitary appended at the end. The overlap of each measured bit string with these filter states is computed, and only measurements exceeding a filter threshold $\xi$ are retained as signal. Although the approximation of the filter state introduces a slight bias, it reliably rejects measurements that are very unlikely to occur under noiseless conditions. Here, we used matrix product states (MPS)~\cite{perezgarcia2007matrix} as the filter state representation, which can be efficiently computed for low enough bond dimensions $\chi$.  We note that alternative representations, such as filter states obtained from  Pauli~\cite{rall2019simulation} or Majorana~\cite{miller2025simulation} propagation, are equally applicable and could potentially be more efficient in some cases. 

The two relevant quantities for the implementation of our MPS-based filtering are the bond dimension $\chi$ of the underlying MPS state and the filtering threshold $\xi$.
In \cref{fig:mps_post_selection_a}, we fix $\xi=10^{-4}$ and vary $\chi$. In all cases, parity violating measurements were removed from the dataset. In the case of only parity-based filtering, we see improvements to the error only occurring after around $\sim10^4$ snapshots. We see an improvement to this when using our filtering technique, but too low bond dimension ($\chi\sim 50-100$) limits the precision. Therefore, a high-enough bond dimension is required for the error mitigation to be successful. In contrast, at $\chi=200$, the one- and two-body local energy estimates approach the accuracy of a statevector filter state, an unbiased but classically intractable reference. In \cref{fig:mps_post_selection_b}, we use a fixed number of 2,000 unitaries, then sweep $\xi$ from $10^{-6}$ to $5\cdot10^{-3}$. We find that higher $\xi$ progressively removes noisy measurements until the filter becomes too aggressive, filtering valid signal measurements and introducing a significant bias. While the ultimate reach of this heuristic error mitigation technique remains to be studied, we found it to produce excellent results for the noisy hardware calculations presented in \cref{fig:evolution_plot_h8_H12}.

\section{Discussion}\label{sec:discussion}

The present work advances QC-AFQMC along the axis of classical post-processing efficiency. 
Our reduction of the classical runtime scaling from $\tilde{\mathcal{O}}(N^{5.5})$ per time step to $\tilde{\mathcal{O}}(N^{4.5})$ brings the 
hybrid algorithm within half a polynomial power of the original $\tilde{\mathcal{O}}(N^{4})$ scaling of ph-AFQMC. Combined with our code optimizations, this substantially increases the accessible system sizes of QC-AFQMC,
as we have shown by executing the post-processing step for up to 20 orbitals (40-qubit systems) with modest computational resources. We estimate that for a 50-orbital system, the classical post-processing would require approximately a week of runtime on ten million CPU cores. For a 100-orbital system, our improvements reduce the estimated classical runtime from close to half a millennium to a couple years using ten million CPU cores. While large, this runtime establishes a credible path toward QC-AFQMC as a potential algorithm for practical quantum advantage once quantum hardware matures.

Our hydrogen-chain results, \ce{H8} using error-mitigated experimental data from IQM Emerald and \ce{H12} with simulated data, show that the walker evolution is stable throughout the imaginary-time propagation when using matchgate shadow tomography and that statistical noise from the estimation process does not adversely impact the AFQMC dynamics. For \ce{H12}, we found the exact estimates and estimates using matchgate shadows to be in agreement, with the latter showing a slightly higher variance which could be reduced by increasing the walker population, as was done in Refs.~\cite{huang2024evaluating,zhao2025quantum}. At 40 qubits, we have successfully run QC-AFQMC at more than double the system size of previous demonstrations. The walker evolution was stable, with local energy estimates using matchgate shadow data agreeing closely with those of classical estimates for all three \ce{Li2O4} configurations, thus allowing us to capture the correct behavior along the reaction pathway. 

Given the success of ph-AFQMC and under the conjecture that quantum hardware can prepare higher-quality trial states than classical methods, these results represent a meaningful step toward executing quantum algorithms at a scale sufficient to tackle industrially relevant problems. However, significant work remains before QC-AFQMC can be incorporated into industrial research workflows. Below we outline several promising directions. 


From the classical compute side, obvious improvements include dedicated low-level optimization and improved parallelization across compute nodes, as well as porting the post-processing code to GPU architectures, which could provide a significant speedup with relatively low implementation overhead~\cite{zhao2025quantum}. Employing a growth estimator to evaluate the ground-state energy could further improve the algorithm's classical scaling by one polynomial power, bringing it down to the cost of a single overlap evaluation. This advantage, however, comes at the cost of increased variance in the energy estimate and the need for smaller time steps, potentially leading to a slower evolution.

One can also use smaller quantum devices to run QC-AFQMC on larger systems~\cite{kiser2025contextual,zhao2025quantum}. In the simplest form, this can be done by preparing quantum trials only within active spaces, while running ph-AFQMC for the full system. Additionally, one could use more advanced methods of compressing active spaces further, such as contextual subspace rotations~\cite{kiser2025contextual}. This way, a trial state of 50 orbitals may already be sufficient for investigating industrially relevant systems. 


From the quantum runtime side, the most straightforward improvement would come from parallelizing circuit execution across multiple devices. We also note that, throughout the analysis, the compilation time of the circuits has been ignored, rendering all estimates somewhat optimistic. 
Minimizing the compilation overhead is an important direction for future work.
One approach would be to exploit the fact that the same trial state is prepared at every time step and leverage the specific static structure of matchgate unitary circuits as discussed in Refs.~\onlinecite{wan2023matchgate,zhao2025quantum}.

Improvements on the shadow tomography methods used in QC-AFQMC would lead to further reductions in both the estimation error and runtime. Orbital-rotated shadows~\cite{zhao2021fermionic} have lower variance compared to matchgate shadows, but come at the cost of higher classical post-processing scaling. Besides studying this trade-off, it is worthwhile investigating whether it is possible to design a shadow protocol which combines the variance of orbital-rotated shadows with the scaling of matchgate shadows. Further, the development of new error mitigation techniques for sampling-based algorithms is critical for pushing the utility of QC-AFQMC in the NISQ as well as early-FT regimes.

A potential resource reduction in both qubits (as a result of lower T-gate requirements for FT implementations) and runtime could be achieved by sampling matchgate circuits in the Clifford group, equivalent to sampling from the signed permutation matrices with positive determinant $\text{Sym}^+(2,2N)$ rather than $O(2N)$~\cite{zhao2021fermionic,wan2023matchgate,heyraud2025unified}. 
Additional improvements could come from utilizing importance sampling within shadow protocols, which offers an underexplored avenue to reduce the runtime. Since walkers evolve only incrementally between time steps, the walker distribution at a given time step carries substantial information about the distribution at the next time step. This knowledge could be exploited to concentrate or inform the sampling of shadow unitaries. 

Here, we have based our quantum FT resource estimates on well established QEC concepts, such as the surface code and magic state distillation routines. State-of-the-art quantum error correction developments such as magic state cultivation~\cite{gidney2024magic} and quantum error correction codes with more efficient encoding rates such as quantum low-density parity-check codes (qLDPC)~\cite{mackay2004sparse} could lead to significant reductions in the physical qubit count and runtime overheads. Another possible reduction is offered by the combination of quantum errror mitigation (QEM) and QEC to lower the resources needed by QEC, an approach that has gained significant attention recently \cite{Zhang_2025_qem,jeon2026quantumerrorcorrectionerrormitigated}. 


Another fundamental aspect of QC-AFQMC requiring further work is centered around determining which quantum methods can prepare more ``powerful" trial states compared to what can be implemented classically. This is particularly important since recent work has shown evidence for cases where higher-fidelity trial states result in lower quality ph-AFQMC~\cite{kjonstad2025systematic,kjønstad2026phaseless}. Developing more efficient methods for loading tensor networks onto quantum computers may be an applicable substitute when the estimated crossover between MSD-AFQMC and QC-AFQMC becomes tangible. 

A closely related challenge is the diminishing overlap problem, whereby the trial state and the SD walkers occupy increasingly orthogonal subspaces as the Hilbert space grows, causing their overlap to decay~\cite{mazzola2022exponential}. Thus, fixed statistical noise in the estimator becomes more dominant. These two problems are coupled since a higher-fidelity trial state naturally has larger overlaps with dominant walkers.  One path forward could be replacing the SD walkers with more complex state representations, such as MPS states~\cite{wouters2014projector,jiang2025unbiasing}.

Finally, while this work has focused on quantum chemistry, it would be of significant interest to extend QC-AFQMC to systems from other domains, such as condensed matter, solid-state and high-energy physics.

\section{Methods}\label{sec:methods}

\subsection{QC-AFQMC}\label{subsec:qcafqmc_description}
The \emph{ab initio} electronic Hamiltonian in second quantization after a Cholesky decomposition is given by \cref{eq:qcham}, and finding the ground state of this system is of broad interest. In the following, we provide a brief overview of the hybrid quantum-classical workflow Quantum-Classical AFQMC (QC-AFQMC). For more details, we refer the reader to Refs.~\onlinecite{zhang2013auxiliary,motta2018ab,huggins2022unbiasing}.

QC-AFQMC performs imaginary-time evolution stochastically with small time steps to reach the ground state of an $N$ spin-orbital chemical system with $\zeta$ fermions via
\begin{equation}
\begin{aligned}
    \left|\Psi_{\text{GS}}\right\rangle& \propto \lim _{\tau \rightarrow \infty} \exp (-\tau \hat{H})\left|\Phi_0\right\rangle\\
    &=\lim _{n\rightarrow \infty}\left(\exp 
 (-\Delta \tau \hat{H})\right)^n|\Phi_0\rangle\,,\label{eq:afqmc}
\end{aligned}
\end{equation}
where \(|\Psi_{\text{GS}}\rangle\) is the true ground state, \(|\Phi_0\rangle\) is an initial state with some overlap with the true ground state \(\langle\Phi_0|\Psi_{\text{GS}}\rangle\neq 0\), and repeated short-imaginary-time propagation implements \(\tau=n\Delta\tau\rightarrow\infty\). 

The Hubbard-Stratonovich transformation~\cite{hubbard1959calculation,*stratonovich1957method,negele2018quantum} expresses the short-time propagator as
\begin{equation}
    e^{-\Delta \tau \hat{H}} = \int \mathrm{d} \mathbf{x}\: p(\mathbf{x}) \hat{B}(\Delta \tau, \mathbf{x})+\mathcal{O}\left(\Delta \tau^2\right)\,,
\end{equation}
where \(p(\mathbf x)\) is Gaussian and \(\hat{B}(\mathbf{x})\) is a one-body propagator coupled to auxiliary fields \(\mathbf x\). By Thouless' theorem~\cite{thouless1960stability,*thouless1961vibrational}, \(\hat{B}(\mathbf{x})\) maps single SDs to single SDs. 

The global wavefunction at time \(\tau\) is
\begin{equation}
|\Psi(\tau)\rangle=\sum_i^{N_\textrm{w}} w_i(\tau) \frac{\left|\varphi_i(\tau)\right\rangle}{\left\langle\Psi_\textrm{T}|\varphi_i(\tau)\right\rangle}
\end{equation}
where for the $i$-th walker of $N_\textrm{w}$ walkers at time $\tau$, $w_i(\tau)$ and $|\varphi_i(\tau)\rangle$ are the walker weight and SD, respectively, and \(|\Psi_\textrm{T}\rangle\) is the trial wavefunction prepared on a quantum device.

Walker propagation and weight updates follow the equations
\begin{align}
\left|\varphi_i(\tau+\Delta \tau)\right\rangle  &=\hat{B}\left(\Delta \tau, \mathbf{x}_i-\overline{\mathbf{x}}_i\right)\left|\varphi_i(\tau)\right\rangle\label{eq:propagate}\\
w_i(\tau+\Delta \tau)  &=I_\text{ph}\left(\mathbf{x}_i, \overline{\mathbf{x}}_i, \tau, \Delta \tau\right) \times w_i(\tau)\,,\label{eq:phaseless}
\end{align}
where \(\overline{\mathbf{x}}_i\) is the force bias \cref{eq:force_bias} for the $i$-th walker and the phaseless importance function is
\begin{equation}\label{eq:importance_function}
    I_{\rm{ph}}\left(\mathbf{x}_i, \overline{\mathbf{x}}_i, \tau, \Delta \tau\right)=\left|s_i(\tau, \Delta \tau) \Gamma_{\mathbf{x}_i}\right|\cdot m\,,
\end{equation}
where 
$\Gamma_{\mathbf{x}_i}=e^{\mathbf{x}_i \cdot \overline{\mathbf{x}}_i-\overline{\mathbf{x}}_i \cdot \overline{\mathbf{x}}_i / 2}$, the overlap ratio
\begin{equation}
    s_i(\tau, \Delta \tau)=\frac{\langle\Psi_{\textrm{T}}|\hat{B}\left(\Delta \tau, \mathbf{x}_i-\overline{\mathbf{x}}_i\right)| \varphi_i(\tau)\rangle}{\left\langle\Psi_{\textrm{T}}|\varphi_i(\tau)\right\rangle}\label{eq:ovlp_ratio}\,,
\end{equation}
and $m=\max \left(0, \cos \left(\arg(s_i(\tau, \Delta \tau))\right)\right)$. \cref{eq:importance_function} constrains random walks with a boundary condition set by the trial wavefunction known as the cosine projection. The weight update method mitigates the phase problem but introduces bias that is eliminated when the trial wavefunction is exactly the true ground state. 

The QC-AFQMC energy at time $\tau$ is
\begin{equation}\label{eq:block_energy}
E_{\rm B}(\tau) = \frac{\sum_iw_i E_{{\rm L},i}}{\sum_i w_i}, \quad E_{{\rm L},i} = 
\frac{\langle \Psi_\text{T}|\hat{H}|\varphi_i(\tau)\rangle}{\langle \Psi_\text{T}|\varphi_i(\tau)\rangle}\,,
\end{equation}
where $E_{{\rm L},i}$ is known as the local energy. 
Without using reblocking techniques for the reduction of autocorrelations between time steps, the estimate of the energy using (QC-/ph-)AFQMC is the average of the block energies $E_{\rm B}(\tau)$ for many discrete time steps $\tau$ beyond an equilibration period. Thus, the energy estimate using (QC-)AFQMC is given by
\begin{equation}
    E=\frac{1}{N_B}\sum_i^{N_B}E_{\rm B}(\tau_i)\,,
\end{equation}
where $\tau_i\in[0,N_B\Delta\tau]$ for $N_B$ many block energy estimates during the time evolution after an equilibration period. 
By the propagation of errors method, the final QC-AFQMC energy error due to the finite measurement under the matchgate shadows protocol is given by~\cite{kiser2024classical}
\begin{equation}
    \varepsilon_{E}=\frac{\varepsilon_{\rm B}}{\sqrt{N_B}}=\frac{\varepsilon_{\rm L}}{\sqrt{N_B}\sqrt{N_\textrm{w}}}\,.
\end{equation}
An important implication of the above relation is that, for a given target variance in the full QC-AFQMC energy calculation $E$, the local energy estimates do not have to achieve the same level of precision. This consequence permits the use of sampling methods to calculate the energy. Further, this relaxation of the precision required of the local energies in turn reduces the number of snapshots required in the matchgate shadows calculations.

We note that the total error of any AFQMC run would also include stochastic errors from the Monte Carlo process. Further, it is worth noting that the errors introduced in the propagation using matchgate shadow data could result in the convergence to an incorrect energy estimate. It has also been shown that covariances exist when using matchgate shadows; thus, the snapshot estimates are not independent~\cite{kiser2024classical}.

\subsection{Calculating overlaps with Slater determinants using matchgate shadows}\label{subsec:calc_ovlp_mg_shadows}
One method for calculating overlaps between a pure quantum state $\ket{\Psi}$ and a classical representation of a SD $\ket{\varphi}$ is through matchgate shadows~\cite{wan2023matchgate}. There are three general steps: state preparation, measurement, and post-processing. In the state preparation step, one prepares 
\begin{equation}\label{eq:tau_state}
    \ket{\Omega}=\frac{\ket{\Psi}+\ket{\boldsymbol{0}}}{\sqrt{2}}\,,
\end{equation}
on a quantum device. In the measurement stage, the state is evolved by a randomly sampled matchgate circuit $U_Q$, which is described by a matrix in the real orthogonal group $Q\in O(2N)$, and performing a computational basis state measurement to obtain the data $\textbf{b}\in\{0,1\}^N$. In the post-processing stage, one calculates an unbiased estimate of the overlap by post-processing the measurement outcome $\textbf{b}$ with the following formula
\begin{equation}\label{eq:unbiased_estimator}
    \hat{o} = \sum_\ell \alpha_{\ell,\zeta,n} q^{(\ell)}(z)/\ell!\,,
\end{equation}
for the polynomial 
\begin{equation}\label{eq:matchgateinversechannel}
    q(z)=\Pf\left[\left.\left(C_{\boldsymbol{0}}+z\tilde C_{\textbf{b}}\right)\right|_{\overline{S}_{\zeta}}\right]\,,
\end{equation}
where $C_{\textbf{b}}$ is the covariance matrix of the measurement outcome $\textbf{b}$, $\tilde C_{\textbf{b}}=W^*\tilde Q Q^\text{T} C_{\textbf{b}}Q\tilde Q^\text{T}W^\dagger$, $\alpha_{\ell,\zeta,n}=2 \binom{2n}{2\ell}\binom{n}{\ell}^{-1}i^{\zeta/2} / 2^{n-\zeta/2}$, $W$ rotates the basis to the set of Majorana operators~\cite{wan2023matchgate}, $\tilde Q$ is the orthogonal matrix representation of the SD $\ket{\varphi}$, $\overline{S}_\zeta:=[2N]\setminus\{1,3,\dots,2\zeta-1\}$, and $[i]=\{1,2,\dots,i\}$.

These unbiased estimates are then averaged, or a median of means is taken, across many measurement runs to obtain an estimate of the overlap. A given accuracy $\varepsilon$ can be reached with $\mathcal{O}(\log(N)\sqrt{N}/\varepsilon^2)$ many samples (\cref{eq:n_samples})~\cite{wan2023matchgate}. 

We note that an alternative, orbital-rotated shadows protocol has been proposed for estimating overlaps with a variance scaling as $\mathcal{O}(\log(N))$~\cite{zhao2021fermionic, kiser2024classical} rather than $\mathcal{O}(\sqrt{N}\log(N))$ in the case of matchgate shadows~\cite{wan2023matchgate}. This, however, comes at the cost of extra qubits, longer circuits and an additional polynomial factor in the post-processing of an overlap---$\mathcal{O}(N^5)$. We have thus chosen to focus solely on the matchgate shadows protocol from Ref.~\onlinecite{wan2023matchgate} in this work.

\subsection{Improving the classical post-processing for calculating the force bias}\label{subsec:fb_improvement}

The calculation of the force bias can be performed at the same cost as an overlap, which was $\mathcal{O}(N^4)$ time~\cite{jiang2025unbiasing}. We now show how it can be expressed as a calculation of a ratio of overlaps and performed in $\mathcal{O}(N^3)$ time using our improved method of classical post-processing. The force bias $\bar{x}_\gamma$ associated with the $\gamma$-th Cholesky operator is defined by the derivative
\begin{equation}
 \bar{x}_\gamma=-\sqrt{\Delta\tau}\left.\frac{\partial\ln(\bra{\psi_{\rm T}}\varphi(\lambda)\rangle)}{\partial\lambda}\right|_{\lambda=0}\,,
\end{equation}
where $|\varphi(\lambda)\rangle=e^{\lambda\hat{v}_\gamma}\ket{\varphi}$ is the Cholesky-evolved SD. Inserting \cref{eq:unbiased_estimator} and \cref{eq:matchgateinversechannel} into the above, 
we get the nested expression
\begin{equation}\label{eq:nested_g_h_force_bias}
     \bar{x}_\gamma= -\sqrt{\Delta\tau}\left.\frac{\partial}{\partial \lambda}\ln\left(g(h(\lambda))\right)\right|_{\lambda=0}\,,
\end{equation}
where we define  $g:\mathbb{C}^{2N\times2N}\rightarrow\mathbb{C}$ as
\begin{equation}
    g(X)=\sum_{\ell=0}^n\frac{\alpha_{\ell,\zeta,n}}{\ell!} \left.\frac{\partial^\ell}{\partial z^\ell}\Pf\!\left[
\left.\left(
C_{\boldsymbol{0}}
+
zX
\right)\right|_{\overline{S}_{\zeta}}
\right]\right|_{z=0}\,,
\end{equation}
Embedded within $h(\lambda)$ is the dependence on the Cholesky matrices. 
In defining $h(\lambda)$, we 
start by defining the function $m:\mathbb{R}\rightarrow\mathbb{C}^{N\times N}$ for the unitary matrix $V$ of the MC walker SD as
\begin{equation}
    m(\lambda) = e^{i\lambda L^{\gamma}}V\,.
\end{equation}
Next, the (special) orthogonal matrix with respect to the SD is given by $\tilde Q:\mathbb{C}^{N\times N}\rightarrow\mathbb{R}^{2N\times 2N}$ as
\begin{equation} \tilde Q(X) = \begin{bmatrix} R_{11} &\dots &R_{1N} \\ 
\vdots &\ddots &\vdots \\ 
R_{N1} &\dots &R_{NN} \end{bmatrix}\,,
\end{equation}
where the blocks are defined by
\begin{equation}
R_{jk} \coloneqq \begin{bmatrix} \mathrm{Re}(X_{jk}) &-\mathrm{Im}(X_{jk}) \\ \mathrm{Im}(X_{jk}) &\mathrm{Re}(X_{jk}) \end{bmatrix}. \end{equation}
Then to calculate the overlap with a particular SD, we have $\tilde C_{\textbf{b}}:\mathbb{R}^{2N\times 2N}\rightarrow \mathbb{C}^{2N\times 2N}$
\begin{equation}
\tilde C_{\textbf{b}}(X)=W^*X Q^\text{T} C_{\textbf{b}}Q X^\text{T}W^\dagger\,.
\end{equation}
With each of these pieces in hand, we can build the expression for $h$ as 
\begin{equation}
\begin{aligned}
h(\lambda)&=\tilde C_{\textbf{b}}(\tilde Q(m(\lambda)))\\
&=W^*\tilde Q(m(\lambda)) Q^\text{T} C_{\textbf{b}}Q \tilde Q(m(\lambda))^\text{T}W^\dagger\,.
\end{aligned}
\end{equation}

Returning to \cref{eq:nested_g_h_force_bias}, we write the force bias terms as calculations of the ratio
\begin{equation}
\begin{aligned}
    \bar{x}_\gamma&=-\sqrt{\Delta\tau} \left.\frac{g'(h(\lambda))}{g(h(\lambda))}\right|_{\lambda=0}\\ &= -\sqrt{\Delta\tau}\frac{\left.g'(h(\lambda))\right|_{\lambda=0}}{g(h(0))}\,,
\end{aligned}
\end{equation}
where the denominator is independent of the Cholesky matrix and is equivalent to the estimation of the overlap with the unevolved MC walker $\langle\psi_{\rm T}|\varphi\rangle$. Thus, it remains only to differentiate the numerator. Applying the chain rule, we obtain
\begin{align}
    \left.\frac{\partial}{\partial \lambda}g(h(\lambda))\right|_{\lambda=0}=&
    \left.\left\langle\overline{\nabla_{h}g},\mathbb{J}_h\right\rangle\right|_{\lambda=0}\,,
\end{align}
where $\langle,\rangle$ indicates the Hilbert-Schmidt inner product.
Since $g$ is a scalar-valued function of a matrix, its derivative with respect to $h$ is actually its gradient. Meanwhile, since $h$ is a function of a scalar to a complex vector space, its derivative with respect to $\lambda$ is its Jacobian  $\mathbb{J}_{h}: \mathbb{R} \to \mathbb{C}^{2N\times 2N}$~\cite{kramer2024tutorialautomaticdifferentiationcomplex}. 
One could straightforwardly apply this for every Cholesky matrix, which would add an $\mathcal{O}(N)$ scaling factor to the calculation of the overlap, but we are targeting a scaling equivalent to the overlap post-processing.  

We can leverage the fact that $\lambda$ multiplies each Cholesky matrix $L^{\gamma}$ to define the auxiliary variable $Z(\lambda)=\lambda L^{\gamma}$ and the auxiliary function $\tilde h(Z(\lambda))=h(\lambda)=\tilde C_{\textbf{b}}(\tilde Q(\tilde m(\lambda L^\gamma)))$ where $\tilde m = e^{iX}V$. We can then find the derivative of $h$ as
\begin{align}
    \left.\frac{\partial h(\lambda)}{\partial \lambda}\right|_{\lambda=0}&=\left.\frac{\partial \tilde h(Z(\lambda))}{\partial\lambda}\right|_{\lambda=0}\\
    &=\mathbb{J}_{\tilde h}(0)\cdot\frac{\partial Z}{\partial\lambda}\\
    &=\mathbb{J}_{\tilde h}(0)\cdot L^{\gamma}\,.
\end{align}
This is the Jacobian-vector product (JVP), which can be seen as a map that takes an input point $\lambda$, at which the Jacobian is evaluated, and a vector which is multiplied by the Jacobian.  In our case, the input would be $\lambda=0$ and the Cholesky matrix $L^{\gamma}$. For simplification of notation, we now label the JVP as the function $S$.
We have now written the $\gamma$-th component of the numerator of the force bias as: 
\begin{equation}
        \left.\frac{\partial}{\partial \lambda}g(h(\lambda))\right|_{\lambda=0} = \left\langle \overline{\nabla_{h}g}, S(L^{\gamma})\right\rangle\,.
\end{equation}
Due to the adjointness properties of linear maps, we can rewrite the expression to have just the Cholesky matrix on the right
\begin{equation}
        \left.\frac{\partial}{\partial \lambda}g(h(\lambda))\right|_{\lambda=0} = \left\langle S^\dagger\left(\overline{\nabla_h g}\right),L^{\gamma}\right\rangle\,,
\end{equation}
where $S^{\dagger}$ is the adjoint of a JVP, also called Vector-Jacobian product (VJP). We can apply the VJP to $\nabla_{h}g$, which is independent of the Cholesky matrices
\begin{equation}
    A\coloneqq S^{\dagger}\left(\overline{\nabla_{h}g}\right) \,.
\end{equation}
So we have that
\begin{equation}
    \left.\frac{\partial}{\partial \lambda}g(h(\lambda))\right|_{\lambda=0} = \sum_{i,j}(A^*_{ij}L^{\gamma}_{ij})\,.
\end{equation}
The final expression is then a trace of a matrix product. The one-time computation of the VJP has the same asymptotic scaling as the underlying function, which is the overlap computation that in our improved method scales as $\mathcal{O}(N^3)$. This final product needs to be evaluated for each Cholesky matrix, and each evaluation scales as $\mathcal{O}(N^2)$, giving a total scaling as $\mathcal{O}(N^3)+\mathcal{O}(N^2N_{\rm C})=\mathcal{O}(N^3)$, yielding the desired scaling improvement. 

\subsection{Review of assumptions on classical runtime scaling estimates}\label{subsec:scaling_assumptions}

Here we discuss the underlying assumptions for the scaling estimates of QC-AFQMC as it relates to existing literature, shown in \cref{fig:complexity_scaling}. The total runtime assumes the time evolution of $1,000$ Monte Carlo walkers over $N^2$ imaginary-time steps. For the QC-AFQMC estimates, the number of snapshots $N_{\rm s}$ is based on the bound from Ref.~\onlinecite{wan2023matchgate}, which scales as $\mathcal{\tilde{O}}(\sqrt{N})$, 
\begin{widetext}
\begin{equation}\label{eq:n_samples}
N_{\rm s}= \lceil9/2\ln(N_\textrm{w}/\delta)\rceil\cdot\lceil24 b(N,\zeta)/\varepsilon^2\rceil\,,
\end{equation}
where $\varepsilon=0.1$, $\delta=0.1$, $N_\textrm{w}=1$ given the results in \cref{fig:snapshot_scaling} that the bound is quite high even for a single walker and~\cite{wan2023matchgate}
\begin{equation} b(N, \zeta) \coloneqq \frac{1}{2^{2N}} \sum_{\substack{\ell_1,\ell_2,\ell_3 \geq 0\\ \ell_1 + \ell_2 + \ell_3 \leq N}}\alpha_{\ell_1,\ell_2,\ell_3}\, \kappa(N, \zeta, \ell_1,\ell_2,\ell_3),  \end{equation}
where 
\begin{equation} \alpha_{\ell_1,\ell_2,\ell_3} \coloneqq \frac{\binom{N}{ \ell_1,\ell_2,\ell_3, N-\ell_1-\ell_2-\ell_3}}{\binom{2N}{ 2\ell_1,2\ell_2,2\ell_3, 2(N-\ell_1-\ell_2-\ell_3)}}\frac{\binom{2N}{ 2(\ell_1 + \ell_3)}}{\binom{N}{ \ell_1 + \ell_3}} \frac{\binom{2N}{ 2(\ell_2 + \ell_3)}}{\binom{N}{ \ell_2 + \ell_3}}\,, \end{equation} 
and
\begin{equation}  \kappa(N,\zeta, \ell_1,\ell_2,\ell_3) \coloneqq 2^{\zeta}\sum_{j=0}^{\zeta/2} \binom{\zeta}{ 2j}\binom{N-\zeta }{ \ell_1 - \zeta/2 +j,\, \ell_2 - \zeta/2 + j, \, \ell_3 - j, \,N - \ell_1 - \ell_2 - \ell_3 -j}\,. \end{equation}
\end{widetext}

We estimate runtime using ten million cores, corresponding to the approximate scale of the world’s largest supercomputers~\cite{top500}. This reflects a best-case scenario that assumes access to maximum available computational resources.

The times for HF-AFQMC are based on running ipie~\cite{malone2023ipie} code for hydrogen chains up to \ce{H54} using a MacBook Pro M1 chip and extrapolating this timing to \ce{H100}. 
The timing for MSD-AFQMC corresponds to the timing for calculating a single determinant of a MSD-AFQMC calculation shown in Figure 9 of Ref.~\onlinecite{jiang2024improved}, multiplied by a conservative bound of $2^{0.37N}$ relevant for highly correlated systems~\cite{kanno2023quantumselectedconfigurationinteractionclassical}.
The original QC-AFQMC~\cite{huggins2022unbiasing} runtime estimate is calculated by timing Clifford shadow code per snapshot up to 30 qubits, extrapolating to the larger system sizes and multiplying by the binomial factor ${\binom{N}{N/2}}^2$. 
The first matchgate implementation of QC-AFQMC scales as $\tilde{\mathcal{O}}(N^{8.5})$ and is based on the reported times from Ref.~\onlinecite{huang2024evaluating} corrected for the fact that they only calculated about ten unique measurements per step~\cite{benchen2024private}. 
The first demonstration of the algorithmic improvements from Ref.~\onlinecite{jiang2025unbiasing} for QC-AFQMC is based on the reported times from Ref.~\onlinecite{zhao2025quantum} and scales as $\tilde{\mathcal{O}}(N^{5.5})$. They used GPUs, and we use the same $50\times$ factor~\cite{zhao2025quantum} to convert it to CPU timing for a comparison with others' and our timings.
Our calculations were based on one walker of a 40 qubit system size using 10,000 unitaries with ten shots each, for a total of 100,000 snapshots, taking 737 seconds for a time step on the LUMI supercomputer (AMD EPYC 7763) using ten CPU cores. Since the energy evaluation and force bias calculations are dependent on the number of Cholesky operators $N_{\rm C}$, we truncate the Cholesky decomposition such that $N_{\rm C}=N$. Our implementation scales as $\tilde{\mathcal{O}}(N^{4.5})$. In each case, we add the time of HF-AFQMC so that the computational time is always longer than HF-AFQMC.

\subsection{Review of assumptions on quantum runtime scaling estimates}\label{subsec:scaling_assumptions_quantum}

For this scaling estimation, we consider a trial state from \cref{subsec:quantum_runtime} and a matchgate circuit sampled from the orthogonal group $O(2N)$. We transpiled this circuit to square-lattice and all-to-all architectures for superconducting- and ion-trap-based quantum computers, respectively. In all cases, the number of snapshots required is based on the same values as the previous section. 

For the NISQ scenario of directly running on physical qubits, we assume a circuit layer to take 50ns and 200$\mu$s for SC and IT, respectively, roughly the timing of a two qubit gate in each technology. On top of the circuits themselves, we assumed that the superconducting devices use 2.5ms of active reset at the start of each circuit with readout and classical post processing taking 1$\mu$s each, and we assumed the ion-trap device requires 200$\mu$s for Doppler cooling, 2ms of sideband cooling, 50$\mu$s for the state preparation and 1ms for readout.
We note that our ion-trap estimate is likely optimistic given the reported runtime of 4 days for calculations on 24 qubits during performance-mode operation, as reported in Ref.~\onlinecite{zhao2025quantum}.

For the fault tolerant runtime estimation, we used the Microsoft quantum resource estimator~\cite{microsoft2024quantum}. Single-qubit and T-gate error rates were set to $10^{-4}$, while two-qubit and measurement error rates were $10^{-3}$. Single- and two-qubit gate times are both 50ns with measurements taking 100ns. We assumed a T-distillation scheme of 15--to--1~\cite{bravyi2005universal} and limited the number of T-factories to half the number of logical qubits required. We also assumed an overall error budget of 10\%, which corresponds to the sum of the logical qubits, the T state error and the rotation gate error for synthesizing arbitrary-rotation gates.

In the case of the modified Hadamard test, we transpiled the circuit for an LUCJ circuit used above and a matchgate circuit constructed from a random unitary matrix that represents a random SD. We used the same assumptions as above for direct and FT operating modes. The number of snapshots $N_{\rm s}$ to get the same $\varepsilon$ and $\delta$ as above is independent of system size and is calculated as $N_{\rm s}=\lceil 2\log(2/\delta)/\varepsilon^2\rceil$ according to Hoeffding's inequality and the Chernoff bound.

\subsection{Details on implementations}\label{subsec:implementation_details}

In this section, we discuss the implementation details for our results in \cref{subsec:qcafqmc_implementation}.
The QC-AFQMC and ph-AFQMC calculations are performed using a version of ipie~\cite{malone2023ipie} that we have modified to incorporate matchgate shadows within the overlap, force bias and local energy calculations with optimizations and algorithmic differentiation performed using JAX~\cite{jax2018github}.
 
We use 800 and 500 walkers for the \ce{H8}- and \ce{H12}-chains, respectively, for 30 blocks of 20 time steps per block with the block energy (\cref{eq:block_energy}) calculated only at the end of each block from the local energies of the walkers. The imaginary time step is set to $\Delta\tau=0.05~{\rm Ha}^{-1}$ for the first 300 time steps and lowered to $0.01~{\rm Ha}^{-1}$ thereafter to decrease the Trotter error. For the matchgate shadows estimation, we used $10^5$ snapshots consisting of $10^4$ unitaries with ten samples (shots) per unitary (See \cref{subsec:multi_shadow} for details on taking multiple shots).

For each of the three \ce{Li2O4} configurations, we evolved 100 walkers for 100 time steps using matchgate shadow data consisting of $4\cdot 10^4$ unitaries with ten shots per unitary. At the end of the 100 time steps, the block energy (\cref{eq:block_energy}) estimates (QC-AFQMC $E_{\rm B}$) were obtained by bootstrap resampling of the local energies of each walker: $2\cdot 10^4$ unitaries were randomly drawn from the original pool of $4\cdot 10^4$ unitaries, and each block energy was recomputed 200 times to determine the statistical error bars. We compared this to running the ph-AFQMC 200 times under the same parameter settings to obtain the block energies (ph-AFQMC $E_{\rm B}$) without the statistical noise from matchgate shadows.

\subsection{Multi-shot Shadow Tomography}\label{subsec:multi_shadow}

In the original shadow protocol presented in Ref.~\onlinecite{huang2020predicting}, only the case of one measurement per circuit is considered. In practice, one might want to collect multiple measurements per sampled shadow unitary. This can be helpful when the incremental time required to rerun a particular circuit on a device is significantly less than preparing and running a completely different circuit. Practically, this can occur when compiling circuits becomes too costly as a proportion of total runtime and limited quantum resources. The performance of the multi-shot shadow tomography has been studied in~\cite{helsen2023thrifty,zhou2023performance}. Multi-shot shadow tomography was proven not to be beneficial in the case of Clifford shadows, while it was shown to be beneficial in the case of Pauli shadows, depending on the observable considered or doped Clifford shadows.

\section*{Acknowledgements}

We thank Emiliano Godinez for his review and feedback on the algorithmic differentiation section of the manuscript. We thank Shiwei Zhang, Lode Pollet and Benchen Huang for helpful discussions early in the project, and Alessio Calzona for his insights into error mitigation techniques. The authors acknowledge CSC, Finland for awarding us access to LUMI, owned by the EuroHPC Joint Undertaking.


\section*{Competing interests}
The authors declare no competing interests.

\appendix

\section{Differentiation method}\label{app:appDalg}
We now summarize the differentiation method for calculating the coefficients of the polynomial generated by a Pfaffian, and for more technical details refer the reader to Ref.~\onlinecite{wan2023matchgate}. 
In the differentiation method, one recursively computes the coefficients of the polynomial defined by $\Pf(A(z))$, where $A(z)$ has dimensions $n\times n$, through the derivatives 
\begin{equation}
    c_\ell = \frac{\partial^\ell\Pf(A(z))|_{z=0}}{\ell!}\,.
\end{equation}
This is apparent because the coefficient $c_\ell$, corresponding to the $\ell$th power of $z$, is precisely the constant term of the polynomial obtained by taking the $\ell$th derivative of the Pfaffian of the matrix $A(z)$ evaluated at $z=0$ and dividing by $\ell!$. The factorial term comes from the process of successively taking the derivative of the function.

Starting with the first power of the derivative of the Pfaffian, one knows that if $A(z)$ is invertible then,
\begin{equation}
    \partial^1 \Pf(A(z)) = \frac12\Pf(A(z))\text{Tr}\left[A^{-1}(z)\partial^1A(z)\right]\,.
\end{equation}
One can split this expression to define the two functions
\begin{align}
    f(z) = \Pf(A(z))&& g(z) = \frac12\text{Tr}[A^{-1}(z)\partial^1A(z)]\,,
\end{align}
and thus $\partial^1f(z)=f(z)g(z)$. Invoking the product rule, one recursively finds for $\ell\geq1$
\begin{equation}
    \partial^\ell f(z) = \sum_{j=0}^{\ell-1}\binom{\ell-1}{j}\left(\partial^{\ell-1-j}f(z)\right)\left(\partial^jg(z)\right)\,.
\end{equation}
For $A(z)=B+zC$, the first derivative is given by $\partial^1A(z)=C$ and  $g(0)=\frac12\text{Tr}\left[B^{-1}C\right]$. Thus, 
\begin{equation}
    \partial^jg(z)|_{z=0}=\frac12(-1)^jj!\text{Tr}\left[\left(B^{-1}C\right)^{j+1}\right]\,.
\end{equation}
The coefficients are recursively calculated as 
\begin{equation}
    \partial^\ell f = \frac12\sum_{j=0}^{\ell-1}(-1)^jj!\binom{\ell-1}{j}\left(\partial^{\ell-1-j}f\right)\text{Tr}\left[\left(B^{-1}C\right)^{j+1}\right]
\end{equation}

This method scales as $\mathcal{O}(n^3)$ from the computation of the eigenvalues of $B^{-1}C$ to calculate the $\text{Tr}\left[\left(B^{-1}C\right)^{i}\right]$ terms and $c_0=\Pf(B)$, both of which scale as $\mathcal{O}(n^3)$.

\section{Symmetry of the Pfaffian Polynomial Coefficients}\label{app:symmetry}
It is possible to exploit a symmetry in the coefficients of the polynomial involved in the overlap estimator to further reduce the runtime of the algorithm. 
Let $B \in SO(2m)$ for $m\in 2\mathbb{Z}^+$, $B = -B^{\rm T}$ and $Q \in U(2m)$ such that $C = QBQ^{\rm T}$. Therefore, $C=-C^{\rm T}$. Define the polynomial
\begin{equation}p(z) = \text{Pf}(B + zC) = \sum_{k=0}^m a_k z^k\,.
\end{equation}

We can then calculate the value of the inverse of $z$ to be
\begin{align}
p\left(\frac{1}{z}\right) &= \text{Pf}\left(B + \frac{1}{z}C\right) \\
&= z^{-m} \text{Pf}(zB + C)\,.
\end{align}

Since $C = QBQ^{\rm T}$, we get that
\begin{align}
\text{Pf}(B + zC) &= \text{Pf}(B + zQBQ^{\rm T}) \\
&= \text{Pf}(Q(-Q^{\rm T}B^{\rm T}Q + zB)Q^{\rm T}) \\
&= \text{Pf}(Q(C + zB)Q^{\rm T}) \\
&= \det(Q) \text{Pf}(C + zB)\,.
\end{align}
Thus, we can get that
\begin{align}
z^{m} p\left(\frac{1}{z}\right) &= \text{Pf}(zB + C) \\
&= \alpha p(z)\,.
\end{align}
Hence,
\begin{align}
\sum_{k=0}^m a_k z^k &= \alpha^{-1} \sum_{k=0}^m a_k z^{m-k}\,,
\end{align}
which shows that the coefficients obey a symmetry when the $C$ matrix is a rotation of the $B$ matrix, which is exactly the case we have when calculating overlaps.

\section{Numerical stability in overlap calculation}\label{app:num_stability}

\begin{figure}
    \centering
    \includegraphics[width=0.92\linewidth]{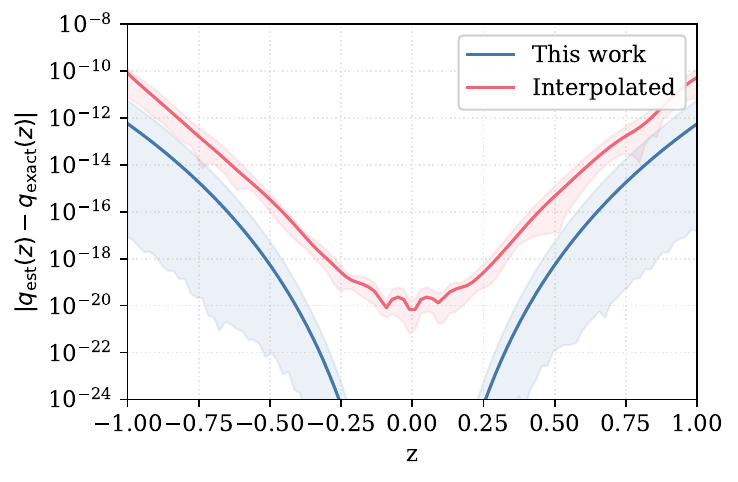}
    \caption{Absolute value of the error between the exact polynomial $q(z)$ and the polynomial reconstructed with coefficients estimated using our new method and interpolation.}
    \label{fig:polynomial_reconstruction}
\end{figure}

In \cref{fig:polynomial_reconstruction}, we show that our proposed method based on Aitken's block transformation of a Pfaffian expression is numerically stable by calculating its error in reconstructing the polynomial $q(z)$ for 40 qubits with a realistic chemical system (\ce{Li2O4}) with respect to the exact values over the interval $z \in [-1,1]$ and comparing to the interpolation method. We see that for various measurement outcomes, we can get better accuracy than the interpolation method. Further analysis on the behavior as it scales to significantly larger systems will be needed.

As described in the main text, our improved overlap computation involves the inversion of a sub-block of a matrix. In principle, nothing guarantees that the sub-block is non-singular, which implies the Pfaffian is 0. In practice, we see that the sub-blocks are indeed typically non-singular and the computation is well-behaved. However, in rare instances, a sub-block is singular or near singular, resulting in Pfaffians close to 0. It was observed that even a few of these occurrences have a very strong impact on the final results. To systematically treat these events, a threshold  $\eta$ is introduced such that if the computed Pfaffian is smaller than the threshold the snapshot is discarded during the computation. We observe that only a small fraction of the available data is discarded at this stage. We find that this fraction increases only slightly with the system size. 
There is tuning required of setting $\eta$ for a given system size, highlighting the heuristic nature of the protocol. Further, the discard rate and accuracy are relatively stable as the system size increases for a given threshold. 

Median of means was mentioned in Ref.~\onlinecite{wan2023matchgate} to calculate the estimate from the snapshot values given in \cref{eq:unbiased_estimator}. In practice, for a distribution that is Gaussian a regular mean is sufficient. We saw that even after the Pfaffian filtering described above, there is still a non-zero likelihood that a significant outlier will appear and worsen the estimate. Thus, we employ a median-of-means strategy, where we carefully select the number of bins. One could also combine this with trimming of outliers for improved performance. Further work should be done to determine a systematic way of setting the parameters for this as we had to manually tune the parameters based on the results of preliminary runs.

As noted in Ref.~\onlinecite{kiser2024classical}, there is some covariance within the snapshot dataset. To mitigate this, we typically use a larger number of unitaries and resample from them when estimating the observables throughout a given timestep. For example for \ce{H8} we have a dataset of 50,000 unitaries, from which 10,000 are uniform-randomly chosen for the estimation. We found this resulted in more consistent and reliable results.

\section{Skew-Tikhonov method}\label{app:skew_method}

Here we maintain the notation presented in \cref{subsec:classical_runtime}. An alternative method to the one based on the use of Aitken's block transformation presented in \cref{subsec:classical_runtime} is to perform a skew-Tikhonov perturbation of the row and column truncated version of $C_{0}$, which results in $B$ in \cref{eq:def_B}, to make it non-singular. However, due to numerical instabilities of this method, we suggest Aitken's block transformation method. 
While the form of $B$ in our case is highly structured and gives a simple formula for the perturbation matrix, this method would work for any non-invertible matrix where the perturbation is added to the dependent columns/rows of the non-invertible matrix. 

Now, we will prove that for a $\beta$ perturbation of $B$, denoted $B'$, such that $B'$ is invertible, the Pfaffian of $A(z)=B+zC$ can be approximated up to accuracy $\mcO(\beta^2)$ by evaluating two Pfaffians of $D(z)=B'+zC$ for $\pm\beta$ to interpolate the value of the unperturbed Pfaffian. This implies the coefficients of the polynomial formed by $D(z)$ will also approach those of $A(z)$, since the coefficients can be interpolated by calculating the Pfaffian for various values of $\beta$. Thus, we can use the differentiation method using $B'$ and $C$ to calculate the coefficients of $\Pf(D(z))$ in $\mcO(n^3)$ that will approximate the coefficients of $\Pf(A(z))$ with accuracy $\mcO(\beta^2)$.

\begin{prop}
    Define the matrix $A(z)=B+zC$ where 
    \begin{equation}
    B:=\bigoplus_{j=1}^{a}
    \begin{bmatrix}
        0&0\\
        0&0
    \end{bmatrix}
    \oplus
    \bigoplus_{j=1}^{b}
    \begin{bmatrix}
        0&1\\
        -1&0
    \end{bmatrix}\,,
    \end{equation}
    for $a+b=n$, $C\in\mbbC^{2n\times 2n}$, 
    \begin{equation}
    E_{\beta}:=\bigoplus_{j=1}^{a}
    \begin{bmatrix}
        0&\beta\\
        -\beta&0
    \end{bmatrix}
    \oplus
    \bigoplus_{j=1}^{b}
    \begin{bmatrix}
        0&0\\
        0&0
    \end{bmatrix}\,,
    \end{equation}
    for $\beta\in\mathbb{R}$, 
    $B_{\pm\beta}'=B+E_{\pm\beta}$ and 
    \begin{equation}\label{eq:perturbed_matrix_fn}
    D_{\pm\beta}(z)=B_{\pm\beta}'+zC\,.
    \end{equation} 
    Then,
    \begin{equation}
        \Pf(D_\pm(z))-\Pf(A(z))\sim\mcO(\beta)\,,
    \end{equation}
    and the coefficients of $\Pf(A(z))$ can be estimated with accuracy at least $\mcO(\beta^2)$ by interpolating between $\Pf(D_{-\beta}(z))$ and $\Pf(D_{+\beta}(z))$.
\end{prop}

\begin{proof}
    Define a modified version of $E_{\beta}$ for $j\in\{1,\dots,a\}$ \begin{equation}
        E_{\beta,j}:=\bigoplus_{i=1}^{j}
    \begin{bmatrix}
        0&0\\
        0&0
    \end{bmatrix}
    \oplus
    \bigoplus_{k=1}^{a-j}
    \begin{bmatrix}
        0&\beta\\
        -\beta&0
    \end{bmatrix}
    \oplus
    \bigoplus_{\ell=1}^{b}
    \begin{bmatrix}
        0&0\\
        0&0
    \end{bmatrix}\,,
    \end{equation}
    where $E_{\beta,0}=E_\beta$ as defined above.

    Now, define $A_s(z)=B+E_{\beta,s}+zC$, where $A_0(z)=D(z)$. Due to the multilinearity of the Pfaffian,
    \begin{align}
        \Pf(D(z))&=\Pf(B+E_{\beta,1}+zC)+\beta\Pf(D^{(1)}(z))\\
        &=\Pf(A_{1}(z))+\beta\Pf(A_0^{(1)}(z))\label{eq:base_case}
    \end{align}
    where $X^{(k)}$ is the $k$th row and $k$th column of $X$ replaced with the $k$th row and $k$th column of $E_1$.

    This is done iteratively $a$ times to get 
    \begin{equation}\label{eq:pfD_step1}
        \Pf(D)=\Pf(A)+\beta\left[\sum_{j=0}^{a-1}\Pf\left(A_j^{(2j+1)}\right)\right]\,,
    \end{equation}
    where we remove the explicit dependence on $z$ for compactness of notation.

    We prove this by induction. The base case is satisfied due to the multilinearity of the Pfaffian, shown in \cref{eq:base_case}. We now assume that this is true for some $k\leq a-1$, 
    \begin{equation}
        \Pf(D)=\Pf(A_k)+\beta\left[\sum_{j=0}^{k-1}\Pf\left(A_j^{(2j+1)}\right)\right]
    \end{equation}
    Due to multilinearity 
    \begin{equation}
        \Pf\left(A_k\right) = \Pf\left(A_{k+1}\right) + \beta\Pf\left(A_{k}^{(2k+1)}\right)\,,
    \end{equation}
    thus, we get as desired for the $k+1$ case,
    \begin{align}
        \Pf(D)=\Pf(A_{k+1})+\beta\left[\sum_{j=0}^{k}\Pf\left(A_j^{(2j+1)}\right)\right]\,.
    \end{align}
    Thus, we get \cref{eq:pfD_step1} as desired.

    We can recognize that for each term in the summand $\Pf(A_j^{(2j+1)})$, there will still be $a-j-1$ many $\beta$ terms in the off-diagonal. We can remove these iteratively to get a summation with a term in the form of $\Pf(A^{(2j+1)})$, using the same method as above. Thus, we find up to second order in $\beta$,
    \begin{equation}
        \Pf(D) = \Pf(A) + \beta\left[\sum_{j=1}^{a}\Pf\left(A^{(2j-1)}\right)\right]+\mcO(\beta^2)\,.
    \end{equation}
    Then, by calculating the Pfaffian twice with $\pm\beta$, we can interpolate the value for $\beta=0$
    \begin{equation}
        \frac{\Pf(D_{-\beta}) + \Pf(D_{+\beta})}{2} = \Pf(A) + \mcO(\beta^2)
    \end{equation}
    
\end{proof}

Thus, we use this method within the differentiation method and calculate the coefficients in time $\mcO(n^3)$.

\section{Absolute energies}\label{app:abs_energies}

Here we provide the absolute energies underlying the results discussed in \cref{subsec:qcafqmc_implementation}. Table~\ref{tab:hchain_energies} reports the absolute energies obtained for the \ce{H8} and \ce{H12} hydrogen chain systems using MG-QC-AFQMC, Exact-QC-AFQMC, and ph-AFQMC, alongside FCI reference values. Table~\ref{tab:energies} presents the total and relative energies for the reactant, transition state (TS), and product configurations of the reaction pathway of \ce{Li2O4}, computed with HF, CCSD and the trial state, as well as the block energy $E_{\rm B}$ at the 100th time step of ph-AFQMC and QC-AFQMC, where the reaction energies $\Delta E_1$ and $\Delta E_2$ are given relative to the reactant and shown in \cref{fig:qc_afqmc_performance_40qb}.

\begin{table}
\centering
\setlength{\tabcolsep}{14pt}
\begin{tabular}{lcc}
\hline\hline
\textbf{Method} & \textbf{\ce{H8} (Ha)} & \textbf{\ce{H12} (Ha)} \\
\hline
MG-QC-AFQMC   & -4.204(05) & -6.307(1)  \\
Exact-QC-AFQMC & -4.204(02) & -6.307(07) \\
ph-AFQMC      & -4.203(02) & -6.307(07) \\
FCI & -4.208 & -6.311 \\
\hline\hline
\end{tabular}
\caption{The absolute energies obtained for the \ce{H8} and \ce{H12} hydrogen chain systems using MG-QC-AFQMC, Exact-QC-AFQMC, and ph-AFQMC, alongside FCI reference values.}
\label{tab:hchain_energies}
\end{table}

\begin{table}
    \resizebox{\columnwidth}{!}{%
\begin{tabular}{lcccccc}
\hline\hline
\textbf{Method} & \textbf{Reactant} \textbf{(Ha)}  & \textbf{TS} \textbf{(Ha)} & \textbf{Product} \textbf{(Ha)} \\
 &  & [$\mathbf{\Delta E_1}$ \textbf{(Ha)}] & [$\mathbf{\Delta E_2}$ \textbf{(Ha)}] \\
\hline
HF               & $-314.097$     & $-314.131$ & $-314.156$  \\
                 &   --             & $-0.034$ & $-0.059$  \\[0.6ex]
CCSD             & $-314.220$     & $-314.203$ & $-314.226$  \\
                 &     --           & $0.017$  & $-0.006$  \\[0.6ex]
Trial state      & $-314.214$     & $-314.194$ & $-314.221$  \\
                 &     --           & $0.019$  & $-0.007$  \\[0.6ex]
ph-AFQMC $E_{\rm B}$  & $-314.219(05)$ & $-314.202(1)$  & $-314.225(02)$ \\
                      &     --           & [$0.016(1)$]   & $-0.006(02)$ \\[0.6ex]
QC-AFQMC $E_{\rm B}$  & $-314.218(03)$ & $-314.203(02)$ & $-314.225(01)$ \\
                      &     --           & $0.015(02)$  & $-0.008(02)$ \\
\hline\hline
\end{tabular}%
}
\caption{Absolute energies (Ha) and reaction energies (Ha) for the methods and configurations along the \ce{Li2O4} rearrangement pathway of \cref{fig:qc_afqmc_performance_40qb}.}
\label{tab:energies}
\end{table}

\section{Matchgate sampling in the lithium-air battery case study}\label{app:large_sampling}

As mentioned in the main text, the trial state is created by taking a linear combination of the top SDs in the configuration-interaction expansion of the CCSD wavefunction. Since we are using 40 qubits, matchgate collection cannot be simulated using statevector simulation. In the matchgate shadows protocol, we have to prepare the state $\ket{\Omega}$, see \cref{eq:tau_state}. Therefore, we have to sample from the equal superposition between the all zeros state and the MSD representation of the trial state. Since SDs are fermionic Gaussian states, we can view this state as a superposition of fermionic Gaussian states. The matchgate shadow data are then collected by repeatedly applying random fermionic Gaussian unitaries and recording the measurement outcomes by sampling the state represented by a linear combination of Gaussian states~\cite{dias2024classical}. While efficient, the sampling process is time-consuming for the full expansion; thus, we truncate the state to reduce the complexity while still retaining the character of the problem. A systematic sweep over the number of retained determinants confirms that three determinants are sufficient to recover the CCSD reaction barrier. The resulting MSD trial wavefunction is stored in the occupation-number representation.

\section{Robust shadow estimation}\label{app:robust}
Since current hardware is typically very noisy, one can implement the so-called robust shadow estimation~\cite{chen2021robust, koh2022classical, zhao2024grouptheoretic, wu2024errormitigated} where the impact of the noise in the channel is approximated by calibrating the protocol with respect to a known state, such as the vacuum state. The error-mitigated expression for the unbiased estimates $\tilde {o}^{(j)}$ uses the noisy channel coefficients $\tilde{f}_{2\ell}$ in
\begin{equation}
    \label{eq:robust_estimation}
    \tilde {o}^{(j)}=2\sum_{\ell=0}^n\tilde{f}_{2\ell}^{-1}c_\ell\,,
\end{equation}
where 
the noisy channel coefficient $\tilde{f}_{2\ell}$ is the coefficient of the $z^\ell$ polynomial 
\begin{equation}
    p_{Q,b}(z)=\binom{n}{\ell}^{-1}\Pf(C_{\mathbf{0}})\Pf\left(-C_{\mathbf{0}}^{-1}+zQ^TC_{b}Q\right)\,.
\end{equation}
The calibration step increases the required number of samples by the same order as the estimation step; thus, the method is still efficient.

Depending on the form of the circuit, the robust shadow protocol can be further refined to include the errors in state preparation~\cite{huang2024evaluating}. This is done by modifying the calibration step to include as many gates in the state preparation as possible such that the vacuum state is unaffected. In our case, this is implemented by removing a single Hadamard gate and a few Pauli-X gates. Thus, the noise from performing the matchgate circuit as well as the noise from the state preparation circuit can be accounted for in the calibration step. We utilize this technique in the experimental results in \cref{subsec:qcafqmc_implementation}.

\bibliography{refs}
\end{document}